\documentclass{article}

\usepackage{./sty/leosstylefile}
\usepackage{./sty/tikzit}
\usetikzlibrary{shapes}
\usetikzlibrary{arrows}

\tikzstyle{node}=[fill=black, draw=black, shape=circle, scale=0.5]
\tikzstyle{wnode}=[fill=white, draw=black, shape=circle, scale=0.5]
\tikzstyle{textbox}=[inner sep=2pt, shape=rectangle, fill=none]
\tikzstyle{textnode}=[inner sep=0mm, shape=circle, fill=white]
\tikzstyle{gnode}=[inner sep=0mm, minimum size=1mm, fill={rgb,255: red,221; green,221; blue,221}, draw={rgb,255: red,221; green,221; blue,221}, shape=circle]
\tikzstyle{refine}=[fill=black, draw=black, shape=regular polygon, regular polygon sides=3, rotate=180, scale=0.5]
\tikzstyle{coarsen}=[fill=white, draw=black, shape=regular polygon, regular polygon sides=3, scale=0.5]
\tikzstyle{bdytextbox}=[fill=white, draw=black, shape=rectangle]
\tikzstyle{redbox}=[fill=white, draw=red, shape=rectangle, text=red]
\tikzstyle{bluecirc}=[inner sep=1mm, fill=white, draw={rgb,255: red,4; green,51; blue,255}, shape=circle, text={rgb,255: red,4; green,51; blue,255}]
\tikzstyle{rednode}=[fill={rgb,255: red,255; green,128; blue,128}, draw={rgb,255: red,255; green,128; blue,128}, shape=circle]
\tikzstyle{new style 0}=[fill=white, draw=red, shape=circle]
\tikzstyle{bluenode}=[fill={rgb,255: red,125; green,221; blue,255}, draw={rgb,255: red,123; green,202; blue,255}, shape=circle]
\tikzstyle{yellownode}=[fill={rgb,255: red,255; green,210; blue,75}, draw={rgb,255: red,255; green,210; blue,75}, shape=circle]
\tikzstyle{blacksq}=[fill=black, draw=black, shape=rectangle, scale=0.5]
\tikzstyle{bluetext}=[fill=none, draw=none, shape=rectangle, text={rgb,255: red,4; green,51; blue,255}]
\tikzstyle{reg}=[draw, fill=white, rounded rectangle, rounded rectangle left arc=none, minimum height=1em, minimum width=1em, node font={\scriptsize}]
\tikzstyle{coreg}=[draw, fill=white, rounded rectangle, rounded rectangle right arc=none, minimum height=1em, minimum width=1em, node font={\scriptsize}]
\tikzstyle{turquoisenode}=[fill={rgb,255: red,115; green,255; blue,239}, draw={rgb,255: red,115; green,255; blue,239}, shape=circle]
\tikzstyle{resistor}=[R]
\tikzstyle{inductor}=[L]
\tikzstyle{capacitor}=[C]
\tikzstyle{voltage-source}=[american voltage source]
\tikzstyle{current-source}=[american current source]

\tikzstyle{edge}=[-, draw=black]
\tikzstyle{diredge}=[->, draw=black]
\tikzstyle{dashed edge}=[-, dashed, dash pattern=on 1pt off 1.5pt, draw=black]
\tikzstyle{dirdash}=[->, dashed, dash pattern=on 2pt off 0.5pt, draw=black]
\tikzstyle{mapsto}=[{|->}, draw=black]
\tikzstyle{gray diredge}=[draw={rgb,255: red,221; green,221; blue,221}, ->]
\tikzstyle{dark grey dirdash}=[->, dashed, dash pattern=on 2pt off 0.5pt, draw={rgb,255: red,81; green,81; blue,81}]
\tikzstyle{doubedge}=[-, draw=black, double=none, double distance=3pt, inner sep=0pt, thick]
\tikzstyle{thedge}=[-, line width=1.5pt, draw=black]
\tikzstyle{gray dashed}=[-, dashed, dash pattern=on 1pt off 1.5pt, draw={rgb,255: red,128; green,128; blue,128}]
\tikzstyle{rededge}=[-, draw=red]
\tikzstyle{gray edge}=[-, draw={rgb,255: red,128; green,128; blue,128}]
\tikzstyle{blthedge}=[-, thick, draw={rgb,255: red,4; green,51; blue,255}]
\tikzstyle{blthdash}=[-, dashed, dash pattern=on 1pt off 1.5pt, thick, draw={rgb,255: red,4; green,51; blue,255}]
\tikzstyle{dirrededge}=[draw=red, ->]
\tikzstyle{tildethrough}=[-, draw=black, decoration={markings,mark=at position 0.5 with {\arrow{Bar}}}, postaction={decorate}]

\graphicspath{ {figures/} }

\title{Disconnection Rules are Complete for Chemical Reactions}
\author{
Ella Gale \footnote{University of Bristol, UK \texttt{ella.gale@bristol.ac.uk}} \and
Leo Lobski \footremember{ucl}{University College London, UK \texttt{leo.lobski.21@ucl.ac.uk f.zanasi@ucl.ac.uk}} \and
Fabio Zanasi \footrecall{ucl} \footnote{University of Bologna, Italy}
}
\date{}

\begin{document}
\maketitle

\begin{abstract}
We provide a category theoretical framework capturing two approaches to graph-based models of chemistry: formal {\em reactions} and {\em disconnection rules}. We model a translation from the latter to the former as a functor, which is faithful, and full up to isomorphism. This allows us to state, as our main result, that the disconnection rules are sound, complete and universal with respect to the reactions. Concretely, this means that every reaction can be decomposed into a sequence of disconnection rules in an essentially unique way. This provides a uniform way to store reaction data, and gives an algorithmic interface between (forward) reaction prediction and (backward) reaction search or retrosynthesis.
\end{abstract}

\section{Introduction}\label{sec:intro}

Graph-based models of chemical processes typically come at two different levels of abstraction: formal {\em reactions} and {\em disconnection rules}. Formal reactions are combinatorial rearrangements of atoms and charge, and are used for reaction prediction and storage of reaction data (see e.g.~the rightmost part of Figure~\ref{fig:clayden} below). Disconnection rules constitute hypothetical bond breaking in the direction opposite to a reaction, and are used for designing synthetic pathways and reaction search, known as {\em retrosynthesis}~\cite{logic-chemical,organic-synthesis,organic-chemistry,compuuter-aided2022,zhong2024recent}  (see e.g.~Figure~\ref{fig:disc-rules} below). Retrosynthetic analysis starts with a target molecule we wish to produce but do not know how. The aim is to ``reduce'' the target molecule to known starting molecules in such a way that when the starting molecules react, the target molecule is obtained as a product. This is done by (formally) partitioning the target molecule into functional parts referred to as {\em synthons}, which are replaced by actual molecules acting as the new targets~\cite{corey1967general,logic-chemical,organic-synthesis,organic-chemistry}. We refer the reader to our previous work on formalising retrosynthesis~\cite{approach-ictac} for the details. In the current work, we are merely interested in the first step of this process: our key observation is that the bond breaking disconnection rules (together with their converse rules), as found in the theory and practise of retrosynthesis (see Figure~\ref{fig:disc-rules}), are, in fact, enough to capture all reactions.
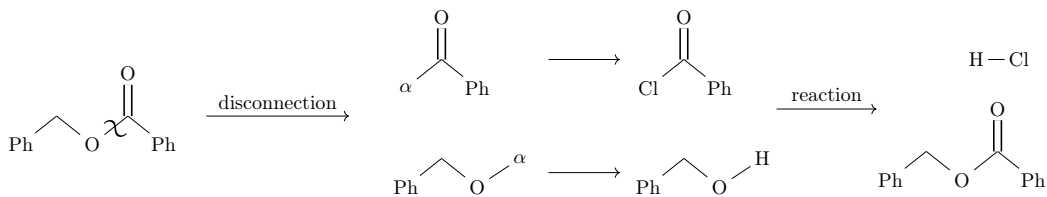
\begin{figure}
\centering
\scalebox{.75}{\begin{tikzpicture}
	\begin{pgfonlayer}{nodelayer}
		\node [style=textbox] (0) at (-14.5, -1) {Ph};
		\node [style=textbox] (1) at (-17, -1) {O};
		\node [style=textbox] (2) at (-19.5, -1) {Ph};
		\node [style=textbox] (3) at (-15.75, 1.5) {O};
		\node [style=none] (5) at (-18.25, 0) {};
		\node [style=none] (6) at (-15.75, 0) {};
		\node [style=textbox] (11) at (-3.5, 1) {Ph};
		\node [style=textbox] (12) at (-3.5, -2.5) {O};
		\node [style=textbox] (13) at (-6, -2.5) {Ph};
		\node [style=textbox] (14) at (-4.75, 3.5) {O};
		\node [style=none] (15) at (-4.75, -1.5) {};
		\node [style=none] (16) at (-4.75, 2) {};
		\node [style=textbox] (17) at (-2, -1.5) {$\alpha$};
		\node [style=textbox] (18) at (-6, 1) {$\alpha$};
		\node [style=textbox] (19) at (5, 1) {Ph};
		\node [style=textbox] (20) at (5, -2.5) {O};
		\node [style=textbox] (21) at (2.5, -2.5) {Ph};
		\node [style=textbox] (22) at (3.75, 3.5) {O};
		\node [style=none] (23) at (3.75, -1.5) {};
		\node [style=none] (24) at (3.75, 2) {};
		\node [style=textbox] (25) at (6.5, -1.5) {H};
		\node [style=textbox] (26) at (2.5, 1) {Cl};
		\node [style=none] (27) at (-1, 2) {};
		\node [style=none] (28) at (1.5, 2) {};
		\node [style=none] (29) at (-1, -2) {};
		\node [style=none] (30) at (1.5, -2) {};
		\node [style=textbox] (46) at (16, -2.25) {Ph};
		\node [style=textbox] (47) at (13.5, -2.25) {O};
		\node [style=textbox] (48) at (11, -2.25) {Ph};
		\node [style=textbox] (49) at (14.75, 0.25) {O};
		\node [style=none] (50) at (12.25, -1.25) {};
		\node [style=none] (51) at (14.75, -1.25) {};
		\node [style=none] (52) at (7, 0.25) {};
		\node [style=none] (53) at (10.5, 0.25) {};
		\node [style=textbox] (54) at (15.5, 2) {Cl};
		\node [style=textbox] (55) at (14, 2) {H};
		\node [style=none] (56) at (-13, 0) {};
		\node [style=none] (57) at (-8, 0) {};
		\node [style=none] (58) at (-10.5, 0.5) {disconnection};
		\node [style=none] (59) at (-16.25, -0.5) {\scalebox{1.8}{\rotatebox{-43}{$\sim$}}};
		\node [style=none] (60) at (8.75, 0.75) {reaction};
	\end{pgfonlayer}
	\begin{pgfonlayer}{edgelayer}
		\draw [style=edge] (2) to (5.center);
		\draw [style=edge] (5.center) to (1);
		\draw [style=edge] (6.center) to (0);
		\draw [style=doubedge] (6.center) to (3);
		\draw [style=edge] (13) to (15.center);
		\draw [style=edge] (15.center) to (12);
		\draw [style=edge] (16.center) to (11);
		\draw [style=doubedge] (16.center) to (14);
		\draw [style=edge] (12) to (17);
		\draw [style=edge] (18) to (16.center);
		\draw [style=edge] (21) to (23.center);
		\draw [style=edge] (23.center) to (20);
		\draw [style=edge] (24.center) to (19);
		\draw [style=doubedge] (24.center) to (22);
		\draw [style=edge] (20) to (25);
		\draw [style=edge] (26) to (24.center);
		\draw [style=diredge] (27.center) to (28.center);
		\draw [style=diredge] (29.center) to (30.center);
		\draw [style=edge] (48) to (50.center);
		\draw [style=edge] (50.center) to (47);
		\draw [style=edge] (47) to (51.center);
		\draw [style=edge] (51.center) to (46);
		\draw [style=doubedge] (51.center) to (49);
		\draw [style=diredge] (52.center) to (53.center);
		\draw [style=edge] (54) to (55);
		\draw [style=diredge] (56.center) to (57.center);
		\draw [style=edge] (1) to (6.center);
	\end{pgfonlayer}
\end{tikzpicture}}
\caption{A simple retrosynthetic sequence. A molecule (far left) is disconnected at the O-COPh bond giving rise to two synthons (left) which can be mapped to precursor molecules (right) which can react to give the product (far right).  \label{fig:clayden}}
\end{figure}
\begin{figure}
\centering
\scalebox{1}{\begin{tikzpicture}
	\begin{pgfonlayer}{nodelayer}
		\node [style=textbox] (1) at (-11, 4.5) {$^u\text{X}^{n}$};
		\node [style=textbox] (8) at (-3.5, 3.25) {$^b\alpha^-$};
		\node [style=none] (9) at (-9, 4.5) {};
		\node [style=none] (10) at (-5, 4.5) {};
		\node [style=none] (11) at (-7, 5) {$E^u_{ab}$};
		\node [style=textbox] (12) at (-2.75, 4.75) {$^u\text{X}^{n+1}$};
		\node [style=textbox] (13) at (-2.25, 6.5) {$^a\alpha$};
		\node [style=none] (14) at (-8.5, 6.5) {$n<0$};
		\node [style=none] (15) at (-8.5, 7.25) {$\text{X}\in\Atset$};
		\node [style=textbox] (16) at (3.5, 5) {$^u\text{X}^{n}$};
		\node [style=textbox] (17) at (12, 3.25) {$^v\alpha^-$};
		\node [style=none] (18) at (5.75, 4.5) {};
		\node [style=none] (19) at (9.75, 4.5) {};
		\node [style=none] (20) at (7.75, 5) {$E^{uv}$};
		\node [style=textbox] (21) at (11.25, 5) {$^u\text{X}^{n+1}$};
		\node [style=none] (23) at (6, 6.5) {$n\geq 0$};
		\node [style=none] (24) at (6, 7.25) {$\text{X}\in\Atset$};
		\node [style=textbox] (25) at (4.25, 3.25) {$^v\alpha$};
		\node [style=textbox] (26) at (-10.75, -2.75) {$^u\text{X}^{n}$};
		\node [style=none] (28) at (-9, -4) {};
		\node [style=none] (29) at (-5, -4) {};
		\node [style=none] (30) at (-7, -3.5) {$I^{uv}$};
		\node [style=textbox] (32) at (-10.75, -5.25) {$^v\text{Y}^{-n}$};
		\node [style=textbox] (33) at (-3.5, -2.75) {$^u\text{X}^{n}$};
		\node [style=textbox] (34) at (-3.5, -5.25) {$^v\text{Y}^{-n}$};
		\node [style=textbox] (35) at (3, -2.75) {$^u\text{X}$};
		\node [style=none] (36) at (5.25, -4) {};
		\node [style=none] (37) at (9.25, -4) {};
		\node [style=none] (38) at (7.25, -3.5) {$C^{uv}_{ab}$};
		\node [style=textbox] (39) at (3, -5.25) {$^v\text{Y}$};
		\node [style=none] (42) at (3.5, -4) {$n$};
		\node [style=none] (43) at (5.75, -2) {$n\notin\{0,\ib\}$};
		\node [style=textbox] (44) at (11, -2.75) {$^u\text{X}$};
		\node [style=textbox] (45) at (11, -5.25) {$^v\text{Y}$};
		\node [style=none] (46) at (12.25, -4) {$n-1$};
		\node [style=textbox] (47) at (13, -2) {$^a\alpha$};
		\node [style=textbox] (48) at (13, -6) {$^b\alpha$};
		\node [style=none] (50) at (5.75, -1.25) {$\text{X, Y}\in\Atset$};
		\node [style=none] (51) at (-7.5, 1.75) {\textit{Electron detachment (negative charge)}};
		\node [style=none] (52) at (7.5, 1.75) {\textit{Electron detachment (nonnegative charge)}};
		\node [style=none] (53) at (-7, -7) {\textit{Ionic bond breaking}};
		\node [style=none] (54) at (7.5, -7) {\textit{Covalent bond breaking}};
	\end{pgfonlayer}
	\begin{pgfonlayer}{edgelayer}
		\draw [style=diredge] (9.center) to (10.center);
		\draw [style=edge] (12) to (13);
		\draw [style=diredge] (18.center) to (19.center);
		\draw [style=edge] (25) to (16);
		\draw [style=diredge] (28.center) to (29.center);
		\draw [style=dashed edge] (26) to (32);
		\draw [style=diredge] (36.center) to (37.center);
		\draw [style=edge] (35) to (39);
		\draw [style=edge] (44) to (45);
		\draw [style=edge] (44) to (47);
		\draw [style=edge] (45) to (48);
	\end{pgfonlayer}
\end{tikzpicture}}
\caption{The four disconnection rules\label{fig:disc-rules}}
\end{figure}
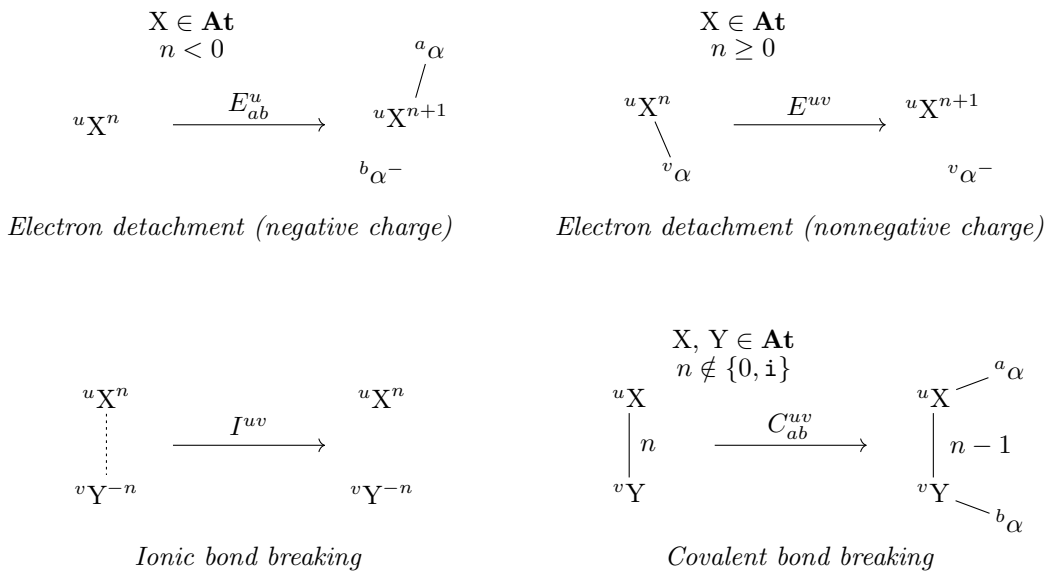

Whereas chemical reactions have been studied formally before, a mathematical description of disconnection rules has received far less attention~\cite{rules-to-term2006,approach-ictac}. Our approach takes a novel perspective on the basic units of retrosynthetic analysis -- the disconnection rules -- by making them first-class citizens of reaction representation. The mathematical and conceptual justification for doing so is the fact that, as we show, both disconnection rules and reactions can be arranged into (monoidal) categories~\cite{approach-ictac}, such that there is a functor taking each sequence of disconnection rules to a reaction. More broadly, our contribution incorporates disconnection rules within the framework of applied category theory~\cite{SevenSketches}, which emphasises compositional modelling as a means to uniformly study systems across various disciplines of science.

The reactions are formalised as certain partial bijections between labelled graphs representing molecules (Definition~\ref{def:category-reactions}). The disconnection rules are defined as partial functions on the labelled graphs, corresponding to the four basic rules in Figure~\ref{fig:disc-rules} (Definition~\ref{def:disc-rules}). The sequences of disconnection rules together with equations between them form the {\em disconnection category} (Definition~\ref{def:disc-cat}), which we think of as syntax for the reactions. This is made precise by constructing a functor from the disconnection category to the category of reactions (Section~\ref{sec:disc-to-react}). Our main results state that the functor is faithful (Theorem~\ref{thm:completeness}) and full up to isomorphism (Theorem~\ref{thm:universality}). Such a categorical perspective provides a precise mathematical meaning to the claim that disconnection rules are sound, complete and universal with respect to the reactions. This implies that every reaction can be decomposed into a sequence of disconnection rules (universality) in an essentially unique way (completeness).

The rest of the paper is structured as follows. Section~\ref{sec:reactions} defines the chemical graphs and the category of reactions. Section~\ref{sec:disc-rules} defines the disconnection rules and their category. It is moreover shown that any sequence of well-typed disconnection rules has a normal form. Section~\ref{sec:disc-to-react} defines a functor from the disconnection category to the category of reactions, and proves completeness and universality. Section~\ref{sec:conclusion} concludes.

\section{Chemical Graphs and Reactions}\label{sec:reactions}

We first define the chemical graphs in Subsection~\ref{subsec:chem-graphs}, which form the objects both in the category of reactions (Definition~\ref{def:category-reactions}) and the disconnection category (Definition~\ref{def:disc-cat}). Chemical reactions are modelled as certain combinatorial rearrangements of chemical graphs that preserve matter and charge in the appropriate sense. The core of the section is the category of reactions (Definition~\ref{def:category-reactions}), the semantic domain for our interpretation of disconnection rules in Section~\ref{sec:disc-to-react}.

\subsection{Chemical Graphs}\label{subsec:chem-graphs}

Let us fix a finite set of {\em vertex labels} $\Atset$, containing the special symbol $\alpha$. Formally, the only assumptions we make about $\Atset$ are (1) it is finite, (2) it contains the special symbol $\alpha$, and (3) $\Atset\setminus\{\alpha\}$ has at least two elements. However, in all the examples we shall assume that $\Atset$ contains a symbol for each main-group element of the periodic table: $\{H,C,O,P,\dots\}\sse\Atset$. For this reason we will also refer to $\Atset$ as the {\em atom labels}. The special symbol $\alpha$ may be thought of as representing an unpaired electron. Similarly, we fix a {\em valence function} $\mathbf v:\Atset\rightarrow\N$ with the only formal assumption that $\mathbf v(\alpha)=1$, but shall assume in the examples that the valence of an element symbol is the number of (unpaired) electrons in its outer electron shell.
\begin{remark}
The reason for choosing such level of generality for the atom labels and their valencies is the ability to model elements which exhibit different valence depending on the context. For instance, one could have separate atom labels for nitrogen whose valence is $5$ (all outer shell electrons are shared or take part in a reaction) or $3$ (two of the outer shell electrons pair with each other).
\end{remark}
Let $\Lab\coloneqq\{0,1,2,3,4,\ib\}$ denote the set of {\em edge labels}, where the integers stand for a covalent bond, and $\ib$ for an ionic bond. We further define maps $\cov,\ion:\Lab\rightarrow\N$: for $\cov$, assign to each edge label $0$, $1$, $2$, $3$, and $4$ the corresponding natural number and let $\ib\mapsto 0$, while for $\ion$, let $0,1,2,3,4\mapsto 0$ and $\ib\mapsto 1$. Finally, let us fix a countable set $\VS$ of {\em vertex names}; we usually denote the elements of $\VS$ by lowercase Latin letters $u,v,w,\dots$.

\begin{definition}[Chemically labelled graph]\label{def:prechemgraph}
A {\em chemically labelled graph} is a triple $(V,\tau,m)$, where $V\sse\VS$ is a finite set of {\em vertices}, $\tau:V\rightarrow\Atset\times\Z$ is a {\em vertex labelling function}, and $m:V\times V\rightarrow\Lab$ is an {\em edge labelling function} satisfying $m(v,v)=0$ and $m(v,w)=m(w,v)$ for all $v,w\in V$.
\end{definition}
Thus, a chemically labelled graph is irreflexive (we interpet the edge label $0$ as no edge) and symmetric, and each of its vertices is labelled with an element of $\Atset$, together with an integer indicating the charge. Given a chemically labelled graph $A$, we write $(V_A,\tau_A,m_A)$ for its vertex set and the labelling functions. We abbreviate the vertex labelling function followed by the first projection as $\tau_A^{\At}$, and similarly we write $\tau_A^{\crg}$ for composition with the second projection.

Given a chemically labelled graph $A$ and vertex names $u,v\in\VS$ such that $u\in V_A$ but $v\notin V_A\setminus\{u\}$, we denote by $A(u\mapsto v)$ the chemically labelled graph whose vertex set is $(V_A\setminus\{u\})\cup\{v\}$, and whose vertex and edge labelling functions agree with those of $A$, treating $v$ as if it were $u$. Further, we define the following special subsets of vertices:
\begin{itemize}
\item {\em $\alpha$-vertices}, whose label is the special symbol: $\alpha(A)\coloneqq\tau_A^{-1}(\alpha,\Z)$,
\item {\em chemical vertices}, whose label is not $\alpha$: $\Chem A\coloneqq V_A\setminus\alpha(A)$,
\item {\em neutral vertices}, whose charge is zero: $\Neu A\coloneqq\tau_A^{-1}(\Atset,0)$,
\item {\em charged vertices}, which have a non-zero charge: $\Crg A\coloneqq V_A\setminus\Neu A$,
\item {\em negative vertices}, which have a negative charge:
$$\Crgn A\coloneqq\{v\in V_A : \tau_A^{\crg}(v)<0\},$$
\item {\em positive vertices}, which have a positive charge:
$$\Crgp A\coloneqq\{v\in V_A : \tau_A^{\crg}(v)>0\}.$$
\end{itemize}
The {\em net charge} of a subset $U\sse V_A$ is the integer $\Net U\coloneqq\sum_{v\in U}\tau_A^{\crg}(v)$.

\begin{example}\label{ex:labelled-graph}
We give three examples of chemically labelled graphs: \textbf{A}, \textbf{B} (carbonate anion) and \textbf{C} (sodium cloride). We adopt the following conventions: (1) the vertex label from $\Atset$ is drawn at the centre of a vertex, (2) the vertex name is drawn as a superscript on the left (so within a single graph, no left superscript appears twice), (3) a non-zero charge is drawn as a superscript on the right, (4) $n$-ary covalent bonds are drawn as $n$ parallel lines, and (5) ionic bonds are drawn as dashed lines.
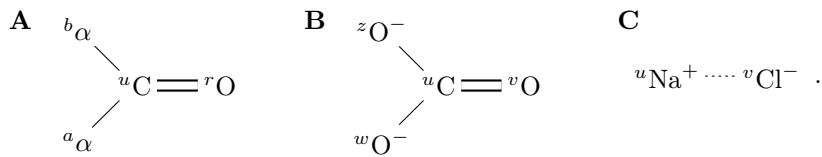
\begin{center}
\scalebox{1}{\begin{tikzpicture}
	\begin{pgfonlayer}{nodelayer}
		\node [style=none] (34) at (-10.25, 1.5) {\textbf{A}};
		\node [style=none] (35) at (-2.5, 1.5) {\textbf{B}};
		\node [style=textbox] (46) at (0.75, -0.25) {$^u\text{C}$};
		\node [style=textbox] (47) at (3, -0.25) {$^v\text{O}$};
		\node [style=textbox] (48) at (-0.75, 1.25) {$^z\text{O}^-$};
		\node [style=textbox] (49) at (-0.75, -1.75) {$^w\text{O}^-$};
		\node [style=textbox] (50) at (-8.75, -1.75) {$^a\alpha$};
		\node [style=textbox] (51) at (-5, -0.25) {$^r\text{O}$};
		\node [style=textbox] (53) at (-8.75, 1.25) {$^b\alpha$};
		\node [style=textbox] (54) at (-7.25, -0.25) {$^u\text{C}$};
		\node [style=textbox] (55) at (6.75, 0) {$^u\text{Na}^+$};
		\node [style=textbox] (56) at (9.5, 0) {$^v\text{Cl}^-$};
		\node [style=none] (57) at (5.75, 1.5) {\textbf{C}};
	\end{pgfonlayer}
	\begin{pgfonlayer}{edgelayer}
		\draw [style=edge] (46) to (49);
		\draw [style=edge] (48) to (46);
		\draw [style=doubedge] (47) to (46);
		\draw [style=edge] (54) to (53);
		\draw [style=edge] (54) to (50);
		\draw [style=doubedge] (51) to (54);
		\draw [style=dashed edge] (55) to (56);
	\end{pgfonlayer}
\end{tikzpicture}
}.
\end{center}
Below we give a table with different kinds of vertex subsets for the graphs:
\begin{center}
\begin{tabular}{ c | c | c | c }
                  & \textbf{A} & \textbf{B} & \textbf{C} \\ \hline
$\alpha$-vertices & $\{a,b\}$ & $\eset$ & $\eset$ \\ \hline
chemical vertices & $\{r,u\}$ & $V_B$ & $V_C$ \\ \hline
neutral vertices & $V_A$ & $\{u,v\}$ & $\eset$ \\ \hline
charged vertices & $\eset$ & $\{w,z\}$ & $V_C$ \\ \hline
negative vertices & $\eset$ & $\{w,z\}$ & $\{v\}$ \\ \hline
positive vertices & $\eset$ & $\eset$ & $\{u\}$ \\ \hline
net charge & $0$ & $-2$ & $0$
\end{tabular}
\end{center}
\end{example}

\begin{definition}[Neighbours]
Given a chemically labelled graph $A$ and a vertex $u\in V_A$, we define the sets of {\em neighbours} $\Nbr_A(u)$, {\em covalent neighbours} $\CN_A(u)$ and {\em ionic neighbours} $\IN_A(u)$ of $u$ as follows:
\begin{align*}
\Nbr_A(u) &\coloneqq\{v\in V_A : m_A(u,v)\neq 0\}, \\
\CN_A(u) &\coloneqq\{v\in V_A : \cov(m_A(u,v))\neq 0\}, \\
\IN_A(u) &\coloneqq\{v\in V_A : \ion(m_A(u,v))\neq 0\}.
\end{align*}
\end{definition}

\begin{definition}[Chemical graph]\label{def:chemgraph}
A {\em chemical graph} $A=(V_A,\tau_A,m_A)$ is a chemically labelled graph satisfying the following additional conditions:
\begin{enumerate}
\item for all $v\in V_A$, we have
$$\left|\tau_A^{\crg}(v)\right|+\sum_{u\in V_A}\cov\left(m_A(u,v)\right) = \mathbf v\tau_A^{\At}(v),$$\label{cgraph:valence}
\item for all $v\in\alpha(A)$ and $w\in V_A$ we have
\begin{enumerate}
\item $\tau_A^{\crg}(v)\in\{-1,0\}$,\label{cgraph:alpha1}
\item $m_A(v,w)\in\{0,1\}$,\label{cgraph:alpha2}
\item $\Nbr_A(u)$ has at most one element, and if $w\in\Nbr_A(u)$, then $w\in\Chem{A}$,\label{cgraph:alpha3}
\end{enumerate}
\item for all $v\in\Chem A$, the set $\IN_A(v)$ has at most one element, and if $u\in\IN_A(v)$, then $\tau_A^{\crg}(v)=-\tau_A^{\crg}(u)$.\label{cgraph:ion}
\end{enumerate}
\end{definition}
Condition~\ref{cgraph:valence} states that the sum of incident covalent bonds together with the absolute value of the charge must equal the valence of the vertex. Conditions~\ref{cgraph:alpha1}-\ref{cgraph:alpha3} say that a vertex labelled by $\alpha$ is either neutral or has charge $-1$, has at most one neighbour, which is chemical and to which it is connected via a single covalent bond. Condition~\ref{cgraph:ion} says that an edge with label $\ib$ only connects charged chemical vertices with equal net charges of opposite signs.

A {\em synthon} is a chemical graph whose set of $\alpha$-vertices is nonempty.
\begin{example}\label{ex:synthon-molecular}
The chemically labelled graphs in Example~\ref{ex:labelled-graph} are, in fact, chemical graphs with the standard valences of the atoms (i.e.~$\mathbf v(\text{C})=4$, $\mathbf v(\text{O})=2$ and $\mathbf v(\text{H})=\mathbf v(\text{Cl})=\mathbf v(\text{Na})=1$). Moreover, \textbf{A} is a synthon.
\end{example}

\subsection{Category of Reactions}\label{subsec:reaction-category}

We define reactions between chemically labelled graphs as partial bijections preserving the atom labels of chemical vertices whose domain and image have the same net charge, with the additional condition that the complements of the domain and image are isomorphic. The intuition is that electrons and atoms' charge may appear and disappear in the course of a reaction in such a way that the overall charge is preserved. We emphasise that these are indeed {\em formal} reactions, in the sense that the only constraints we impose are preservation of matter and charge: in order to capture the chemically feasible reactions, further constraints are needed either in the form of empirical data, or introduction of kinematics or dynamics into the model.

This way of representing reactions is motivated by and formally connected to {\em double pushout graph rewriting}~\cite{inferring-rule-composition,intermediate-level,rewriting-life21}: in fact, every reaction can be represented as a double pushout diagram in the category of chemical graphs.

\begin{definition}[Category of reactions]\label{def:category-reactions}
We denote by $\React$ the {\em category of reactions}, whose
\begin{itemize}
\item objects are chemical graphs,
\item morphisms $A\rightarrow B$ are tuples $(U_A,U_B,b,i)$, where
\begin{itemize}
\item $U_A\sse V_A$ and $U_B\sse V_B$ are subsets with $\Net{U_A}=\Net{U_B}$,
\item $b:\Chem{U_A}\rightarrow\Chem{U_B}$ is a bijection preserving the atom labels,
\item $i:V_A\setminus U_A\rightarrow V_B\setminus U_B$ is an isomorphism of labelled graphs,
\end{itemize}
such that for all $u\in\Chem{U_A}$ and $a\in V_A\setminus U_A$ we have
$$m_A(u,a)=m_B(bu,ia),$$
\item the composition of $(U_A,U_B,b,i):A\rightarrow B$ and $(W_B,W_C,c,j):B\rightarrow C$ is
$$(Z_A,Z_C,(c+j)(b+i),ji):A\rightarrow C,$$
where $Z_A\coloneqq U_A\cup i^{-1}(W_B\cap (V_B\setminus U_B))$ and $Z_C\coloneqq W_C\cup j(U_B\cap (V_B\setminus W_B))$,
\item for a chemical graph $A$, the identity is given by $(\eset,\eset,!,\id_A)$, where $!$ is the unique endomorphism on the empty set.
\end{itemize}
\end{definition}
Note that the composition in $\React$ is {\em not} the composition in the usual category of partial bijections: instead, it crucially relies on the fact that there is an isomorphism between the unchanged parts of the graph. The category $\React$ has a dagger structure~\cite{selinger-dagger,heunen-vicary-book}: the dagger of $(U_A,U_B,b,i):A\rightarrow B$ is given by $(U_B,U_A,b^{-1},i^{-1}):B\rightarrow A$. The dagger of $r\in\React$ is denoted by $\overline r$.

\begin{example}\label{ex:reaction}
We give an example of a reaction (formation of benzyl benzoate from benzoyl chloride and benzyl alcohol) below, where both $b$ and $i$ are identity maps. Here {Ph} stands for the phenyl group, and we use the convention from chemistry where an unlabelled vertex is a carbon atom with an appropriate number of hydrogen atoms attached.
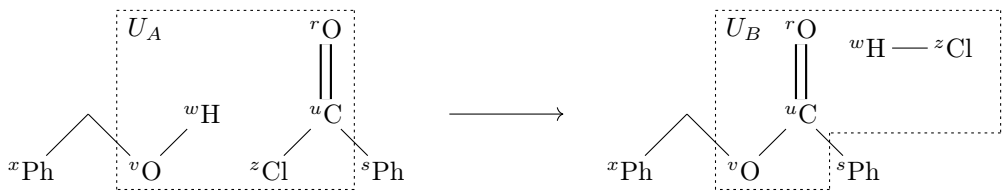
\begin{center}
\scalebox{1}{\begin{tikzpicture}
	\begin{pgfonlayer}{nodelayer}
		\node [style=textbox] (1) at (-9.5, -1.5) {$^v\text{O}$};
		\node [style=textbox] (2) at (-12.5, -1.5) {$^x\text{Ph}$};
		\node [style=textbox] (6) at (-8, 0) {$^w\text{H}$};
		\node [style=none] (14) at (-1.5, 0) {};
		\node [style=none] (15) at (1.5, 0) {};
		\node [style=textbox] (18) at (-3.25, -1.5) {$^s\text{Ph}$};
		\node [style=textbox] (19) at (-4.75, 2.25) {$^r\text{O}$};
		\node [style=textbox] (20) at (-6.25, -1.5) {$^z\text{Cl}$};
		\node [style=textbox] (21) at (-4.75, 0) {$^u\text{C}$};
		\node [style=none] (25) at (-4, -2) {};
		\node [style=none] (26) at (-4, 2.75) {};
		\node [style=none] (27) at (-9.5, 2.25) {$U_A$};
		\node [style=none] (41) at (-11, 0) {};
		\node [style=none] (23) at (-10.25, 2.75) {};
		\node [style=none] (24) at (-10.25, -2) {};
		\node [style=textbox] (42) at (6.25, -1.5) {$^v\text{O}$};
		\node [style=textbox] (43) at (3.25, -1.5) {$^x\text{Ph}$};
		\node [style=textbox] (44) at (9.5, 1.75) {$^w\text{H}$};
		\node [style=textbox] (45) at (9.25, -1.5) {$^s\text{Ph}$};
		\node [style=textbox] (46) at (7.75, 2.25) {$^r\text{O}$};
		\node [style=textbox] (47) at (11.75, 1.75) {$^z\text{Cl}$};
		\node [style=textbox] (48) at (7.75, 0) {$^u\text{C}$};
		\node [style=none] (49) at (8.5, -2) {};
		\node [style=none] (51) at (6.25, 2.25) {$U_B$};
		\node [style=none] (52) at (4.75, 0) {};
		\node [style=none] (53) at (5.5, 2.75) {};
		\node [style=none] (54) at (5.5, -2) {};
		\node [style=none] (55) at (8.5, -0.5) {};
		\node [style=none] (56) at (13, -0.5) {};
		\node [style=none] (57) at (13, 2.75) {};
	\end{pgfonlayer}
	\begin{pgfonlayer}{edgelayer}
		\draw [style=edge] (1) to (6);
		\draw [style=diredge] (14.center) to (15.center);
		\draw [style=edge] (21) to (20);
		\draw [style=edge] (21) to (18);
		\draw [style=doubedge] (19) to (21);
		\draw [style=dashed edge] (23.center) to (26.center);
		\draw [style=dashed edge] (26.center) to (25.center);
		\draw [style=dashed edge] (25.center) to (24.center);
		\draw [style=edge] (2) to (41.center);
		\draw [style=edge] (41.center) to (1);
		\draw [style=dashed edge] (24.center) to (23.center);
		\draw [style=edge] (48) to (45);
		\draw [style=doubedge] (46) to (48);
		\draw [style=dashed edge] (49.center) to (54.center);
		\draw [style=edge] (43) to (52.center);
		\draw [style=edge] (52.center) to (42);
		\draw [style=dashed edge] (54.center) to (53.center);
		\draw [style=edge] (44) to (47);
		\draw [style=edge] (48) to (42);
		\draw [style=dashed edge] (53.center) to (57.center);
		\draw [style=dashed edge] (57.center) to (56.center);
		\draw [style=dashed edge] (56.center) to (55.center);
		\draw [style=dashed edge] (55.center) to (49.center);
	\end{pgfonlayer}
\end{tikzpicture}
}.
\end{center}
\end{example}

\section{Disconnection Rules}\label{sec:disc-rules}

A {\em disconnection rule} is a partial endofunction on the set of chemical graphs. We define four classes of disconnection rules, all of which have a clear chemical significance: two versions of {\em electron detachment}, {\em ionic bond breaking} and {\em covalent bond breaking}. The reader may want to return to Figure~\ref{fig:disc-rules} in Section~\ref{sec:intro} before Definition~\ref{def:disc-rules} below, as it gives an intuitive explanation of our approach.

For the purposes of mathematical precision, our set of four disconnection rules is more fine-grained than what one would see in a typical textbook on retrosynthesis, where movement of electrons is usually implicitly modelled in the same step as disconnecting a bond, rather than including electron detachment as a separate step (see, for instance, the discussion on the choice of polarity in~\cite[p.~9]{organic-synthesis}).

We treat the disconnection rules as syntax, which generate the {\em terms} (Definition~\ref{def:terms}), whose equivalence classes under the equations of Figure~\ref{fig:disc-axioms} form the morphisms in the {\em disconnection category} (Definition~\ref{def:disc-cat}). The payoff of such a syntactic presentation is an axiomatic view of chemical reactions: in Section~\ref{sec:disc-to-react}, we construct a functor from the disconnection category to the category of reactions, and show that every reaction can be represented as a sequence of disconnection rules in an essentially unique way.

\begin{definition}[Disconnection rules]\label{def:disc-rules}
Let $u,v,a,b\in\VS$ be pairwise distinct vertex names. Let $U\in\{u,uv\}$ and $D\in\{\eset,ab\}$ range over the specified lists of vertex names. The four {\em disconnection rules} are partial functions on the set chemical graphs, defined by the table in Figure~\ref{fig:disc-rules-partial-fns} as follows: a chemical graph $A$ is in the domain of $d^U_D$ if $U\sse V_A$ but $D\cap V_A=\eset$, and the additional conditions of the first column hold; the output chemical graph $d(A)$ has the vertex set $V_A\cup D$, and the labelling functions on $U\cup D$ are defined by the remaining columns, while the labelling functions agree with those of $A$ on $V_A\setminus U$.
\end{definition}
\begin{figure}[h]
\centering
\makebox[\textwidth][c]{%
\begin{tabular}{| c || c | c | c | c | c | c | c | c |}
\hline
$d^U_D$ & $A\in\dom(d)$ & $\tau^{\crg}_{d(A)}(u)$ & $\tau^{\crg}_{d(A)}(v)$ & $\tau_{d(A)}(a)$ & $\tau_{d(A)}(b)$ & $m_{d(A)}(u,v)$ & $m_{d(A)}(u,a)$ & $m_{d(A)}(v,b)$ \rule{0pt}{10pt}\rule[-5pt]{0pt}{0pt} \\ \hhline{|=#=|=|=|=|=|=|=|=|}
$E^u_{ab}$ & \makecell{$u\in\Chem A$ \\ $u\in\Crgn A$} & $\tau_A^{\crg}(u)+1$ & N/A & $(\alpha,0)$ & $(\alpha,-1)$ & N/A & $1$ & N/A \\ \hline
$E^{uv}$ & \makecell{$u\in\Chem A$ \\ $u\notin\Crgn A$ \\ $v\in\alpha(A)$ \\ $m_A(u,v)=1$} & $\tau_A^{\crg}(u)+1$ & $-1$ & N/A & N/A & $0$ & N/A & N/A \\ \hline
$I^{uv}$ & \makecell{$m_A(u,v)=\ib$ \\ $u\in\Crgp A$ \\ $v\in\Crgn A$} & $\tau_A^{\crg}(u)$ & $\tau_A^{\crg}(v)$ & N/A & N/A & $0$ & N/A & N/A \\ \hline
$C^{uv}_{ab}$ & \makecell{$u,v\in\Chem A$ \\ $m_A(u,v)\notin\{0,\ib\}$} & $\tau_A^{\crg}(u)$ & $\tau_A^{\crg}(v)$ & $(\alpha,0)$ & $(\alpha,0)$ & $m_A(u,v)-1$ & $1$ & $1$ \\ \hline
\end{tabular}
}%
\caption{The disconnection rules defined as partial functions\label{fig:disc-rules-partial-fns}}
\end{figure}
Note that the disconnection rules look a lot like (a subset of) morphisms in $\React$, except that we keep track of the precise vertex names, and a rule applies to a whole set of chemical graphs. We make this connection precise in Section~\ref{sec:disc-to-react}. We further observe that each disconnection rule is injective (as a partial function), and hence has an inverse partial function.

We use the disconnection rules to define the {\em terms}, which will be used to define the morphisms in the disconnection category.
\begin{definition}[Terms]\label{def:terms}
The set of {\em terms} with types is generated by the following recursive procedure:
\begin{itemize}
\item for every chemical graph $A$, let $\id:A\rightarrow A$ be a term,
\item for every chemical graph $A$ and every $u\in V_A$, let $S^u:A\rightarrow A$ be a term,
\item for every chemical graph $A$, every $u\in\alpha(A)$ and every $v\in\VS$ such that $v\notin V_A\setminus\{u\}$, let $R^{u\mapsto v}:A\rightarrow A(u\mapsto v)$ be a term,
\item for every disconnection rule $d$ and every chemical graph $A$ in the domain of $d$, both $d:A\rightarrow d(A)$ and $\bar d:d(A)\rightarrow A$ are terms,
\item if $\mathtt t:A\rightarrow B$ and $\mathtt s:B\rightarrow C$ are terms, then $\mathtt t;\mathtt s:A\rightarrow C$ is a term.
\end{itemize}
\end{definition}
The first and the fifth items take care of the usual categorical structure, while the terms $S^u$ generated by the second item correspond to ``touching'' the vertex $u$ without changing the structure of the graph, and the terms $R^{u\mapsto v}$ rename an existing vertex $u$ into a fresh vertex $v$.

We refer to the terms of the form $d:A\rightarrow B$ and $\bar d:B\rightarrow A$ generated by the fourth item as {\em disconnections} and {\em connections}, respectively. More specifically, we use the symbols $E^{<0}$, $E^{\geq 0}$, $I$ and $C$ to denote the disconnections corresponding to the specific disconnection rules, and similarly the symbols $\bar{E}^{<0}$, $\bar{E}^{\geq 0}$, $\bar{I}$ and $\bar{C}$ refer to the corresponding connections. Similarly, $S$ and $R$ refer to the terms generated by the second and third items. The same letters in the typewriter type font ($\mathtt E^{<0}$, $\mathtt E^{\geq 0}$, $\mathtt I$, $\mathtt C$, $\bar{\mathtt E}^{<0}$, $\bar{\mathtt E}^{\geq 0}$, $\bar{\mathtt I}$, $\mathtt S$ and $\mathtt R$) are used to denote a sequence of terms of the corresponding kind.

Let us define the endofunction $\overline{()}$ on terms by the following recursion:
\begin{itemize}
\item $\left(\id:A\rightarrow A\right) \mapsto \left(\id:A\rightarrow A\right)$,
\item $\left(S^u:A\rightarrow A\right) \mapsto \left(S^u:A\rightarrow A\right)$,
\item $\left(R^{u\mapsto v}:A\rightarrow A(u\mapsto v)\right) \mapsto \left(R^{v\mapsto u}:A(u\mapsto v)\rightarrow A\right)$,
\item $\left(d:A\rightarrow B\right) \mapsto \left(\bar d:B\rightarrow A\right)$,
\item $\left(\bar d:A\rightarrow B\right) \mapsto \left(d:B\rightarrow A\right)$,
\item $\overline{\mathtt t;\mathtt s}\coloneqq\overline{\mathtt s};\overline{\mathtt t}$.
\end{itemize}

For defining equations, it will be useful to allow untyped terms: the equations (Figure~\ref{fig:disc-axioms}) capture interactions between local graph transformations (i.e.~the disconnection rules), so that the same equation should hold for a whole class of chemical graphs.
\begin{definition}[Untyped terms, well-typedness]\label{def:untyped-terms}
An {\em untyped term} is an element of the free monoid on the set
$$\{\id,S^u,R^{a\mapsto b},E^{ua},E^u_{ab},C^{uv}_{ab},I^{uv},\bar E^{ua},\bar E^u_{ab},\bar C^{uv}_{ab},\bar I^{uv} : u,v,a,b\in\VS\},$$
where we use the symbol $;$ to indicate the multiplication of the monoid.

Given an untyped term $\mathtt t$ and chemical graphs $A$ and $B$, we say that the expression $\mathtt t : A\rightarrow B$
is {\em well-typed} if it is in fact a term, that is, if it can be constructed using the recursive procedure of Definition~\ref{def:terms}.
\end{definition}
We define the binary relation $\leq$ on the set of untyped terms by letting $\mathtt t\leq\mathtt s$ if whenever $\mathtt t : A\rightarrow B$ is well-typed, then so is $\mathtt s : A\rightarrow B$.

The endofunction $\overline{()}$ on the untyped terms is defined in exactly the same way as for the terms with types, simply ignoring the types. Note that $\mathtt t : A\rightarrow B$ is well-typed if and only if $\overline{\mathtt t} : B\rightarrow A$ is. Moreover, observe that $\leq$ defines a preorder on the untyped terms. Consequently, we have $\mathtt t\leq\mathtt s$ if and only if $\overline{\mathtt t}\leq\overline{\mathtt s}$.

Given an untyped term $\mathtt t$, there are either no chemical graphs such that $\mathtt t : A\rightarrow B$ is well-typed, or there are infinitely many such graphs. The latter case is the reason for introducing the untyped terms: we want certain equalities to hold {\em whenever} both sides are well-typed.
\begin{definition}[Term equality]\label{def:term-id}
Let $\approx$ be an equivalence relation on the set of untyped terms. This induces the equivalence relation $\equiv$ on the set of terms as follows: for two terms $\mathtt t,\mathtt s:A\rightarrow B$ with the same type, we let $\mathtt t\equiv\mathtt s$ if either $\mathtt t\approx\mathtt s$ or $\overline{\mathtt t}\approx\overline{\mathtt s}$ as untyped terms.
\end{definition}
Given an equivalence relation $\approx$ on the untyped terms, let us introduce the following shorthand binary relations on the untyped terms:
\begin{itemize}
\item $\mathtt t\lesssim\mathtt s$ if $\mathtt t\approx\mathtt s$ and $\mathtt t\leq\mathtt s$,
\item $\mathtt t\simeq\mathtt s$ if $\mathtt t\lesssim\mathtt s$ and $\mathtt s\lesssim\mathtt t$.
\end{itemize}

\begin{definition}[Disconnection category]\label{def:disc-cat}
The {\em disconnection category} $\Disc$ has as objects the chemical graphs. The set of morphisms $\Disc(A,B)$ is given by the terms of type $A\rightarrow B$, subject to the usual associativity and unitality equations of a category, together with the identities $\equiv$ induced (in the sense of Definition~\ref{def:term-id}) by the equivalence relation defined in Figure~\ref{fig:disc-axioms}.
\end{definition}
\begin{figure}[p]
\centering
\begin{minipage}{0.5\textwidth}
\fbox{
\begin{minipage}{\textwidth}
\begin{align}
R^{u\mapsto z};R^{z\mapsto w} &\lesssim R^{u\mapsto w}\label{disc-eq:trans} \\
R^{u\mapsto z};R^{v\mapsto w} &\approx R^{v\mapsto w};R^{u\mapsto z}\label{disc-eq:rcomm} \\
R^{u\mapsto u} &\lesssim S^u\label{disc-eq:refl} \\
R^{b\mapsto z};R^{a\mapsto b} &\approx S^b;R^{a\mapsto z} \label{disc-eq:rsymm} \\
R^{u\mapsto v};S^w &\approx S^w;R^{u\mapsto v}\label{disc-eq:sr1} \\
R^{u\mapsto v};S^v &\simeq S^u;R^{u\mapsto v} \simeq R^{u\mapsto v}\label{disc-eq:sr2}
\end{align}
\end{minipage}}\vspace{0.2cm}
\fbox{
\begin{minipage}{\textwidth}
\begin{align}
R^{u\mapsto v};d^U_D &\approx d^U_D;R^{u\mapsto v}\label{disc-eq:rd1} \\
R^{u\mapsto v};E^{wv} &\simeq E^{wu};R^{u\mapsto v}\label{disc-eq:rd2} \\
d^U_{D[u]};R^{u\mapsto v} &\simeq d^U_{D[v/u]}\label{disc-eq:rd3} \\
d^{U'}_{ij};\bar h^U_{ab};R^{i\mapsto c};R^{j\mapsto d} &\lesssim \bar h^U_{ab};d^{U'}_{cd}\label{disc-eq:rd4}
\end{align}
\end{minipage}}\vspace{0.2cm}
\fbox{
\begin{minipage}{\textwidth}
\begin{align}
d^U_{ab};\bar d^U_{cd} &\approx S^U;R^{c\mapsto a};R^{d\mapsto b}\label{disc-eq:ddbar1} \\
d^U_{ab};\bar d^U_{cb} &\approx S^U;R^{c\mapsto a}\label{disc-eq:ddbar2} \\
d^U_{ab};\bar d^U_{ad} &\approx S^U;R^{d\mapsto b}\label{disc-eq:ddbar3} \\
d^U_D;\bar d^U_D &\lesssim S^U\label{disc-eq:ddbar4} \\
\bar d^U_D;d^U_D &\lesssim S^U;S^D\label{disc-eq:ddbar4-2} \\
E^{ua};\bar E^{ub} &\approx S^u;R^{a\mapsto z};R^{b\mapsto a};R^{z\mapsto b}\label{disc-eq:eebar} \\
\bar d^{uv};d^{wz} &\approx d^{wz};\bar d^{uv}\label{disc-eq:comm2}
\end{align}
\end{minipage}}%
\end{minipage}%
\begin{minipage}{0.07\textwidth}
\hspace{0.1cm}
\end{minipage}%
\begin{minipage}{0.4\textwidth}
\fbox{
\begin{minipage}{\textwidth}
\begin{align}
S^u;S^v &\simeq S^v;S^u\label{disc-eq:scomm} \\
S^u;S^u &\simeq S^u\label{disc-eq:sidem} \\
S^u;d^U_D &\lesssim d^U_D;S^u\label{disc-eq:sd1} \\
d^{U[v]}_D;S^v &\simeq d^{U[v]}_D\label{disc-eq:sd2} \\
C^{uv}_{ab} &\simeq C^{vu}_{ba}\label{disc-eq:cs}
\end{align}
\end{minipage}}\vspace{0.2cm}
\fbox{
\begin{minipage}{\textwidth}
\begin{align}
d^U_D;d^{U'}_{D'} &\simeq d^{U'}_{D'};d^U_D\label{disc-eq:comm1} \\
C^{uv}_{ab};I^{wz} &\simeq I^{wz};C^{uv}_{ab}\label{disc-eq:comm3} \\
E^u_{ab};I^{wz} &\lesssim I^{wz};E^u_{ab}\label{disc-eq:comm4} \\
E^{uv};I^{wz} &\lesssim I^{wz};E^{uv}\label{disc-eq:comm5} \\
\bar E^{uv};I^{wz} &\lesssim I^{wz};\bar E^{uv}\label{disc-eq:comm6} \\
\bar E^u_{ab};I^{wz} &\lesssim I^{wz};\bar E^u_{ab}\label{disc-eq:comm7} \\
\bar C^{uv}_{ab};I^{wz} &\lesssim I^{wz};\bar C^{uv}_{ab}\label{disc-eq:comm8} \\
E^u_{ab};C^{wz}_{cd} &\simeq C^{wz}_{cd};E^u_{ab}\label{disc-eq:comm9} \\
E^{uv};C^{wz}_{cd} &\lesssim C^{wz}_{cd};E^{uv}\label{disc-eq:comm10} \\
\bar E^{uv};C^{wz}_{cd} &\simeq C^{wz}_{cd};\bar E^{uv}\label{disc-eq:comm11} \\
E^{uv};E^w_{cd} &\lesssim E^w_{cd};E^{uv}\label{disc-eq:comm12} \\
\bar E^{uv};E^w_{cd} &\simeq E^w_{cd};\bar E^{uv}\label{disc-eq:comm13}
\end{align}
\end{minipage}}
\end{minipage}
\caption{The equivalence relation $\approx$ inducing the identities in the disconnection category. Here $d$ and $h$ range over $\{E,C,I\}$, while $S^U$ stands for the sequence $S^u;S^w$ if $U=uw$. Given vertex names $a,b\in\VS$, the notation $D[a]$ means $a$ occurs in $D$, and $D[b/a]$ means the occurrence of $a$ in $D$ is replaced with $b$. Note that we use the shorthand relations $\lesssim$ and $\simeq$: these are strictly speaking not part of the definition, but we use them to provide the extra information of when well-typedness of one side of an identity implies well-typedness of the other.\label{fig:disc-axioms}}
\end{figure}

Note that the assignment $\overline{()}:\Disc\rightarrow\Disc$ is functorial, thus making $\Disc$ a dagger category~\cite{selinger-dagger,heunen-vicary-book}.

\begin{proposition}\label{prop:disc-ids}
The following identities are derivable in $\Disc$:
\begin{align}
d^U_{D[a]};S^a &\simeq d^U_{D[a]},\label{prop:disc-sabsorb} \\
\bar d^U_{ab};d^U_{cd} &\approx S^U;R^{a\mapsto j};R^{b\mapsto d};R^{j\mapsto c},\label{prop:disc-ids0} \\
R^{z\mapsto c};R^{w\mapsto d};d^U_{ab} &\approx R^{z\mapsto a};R^{w\mapsto b};d^U_{cd},\label{prop:disc-ids1} \\
R^{z\mapsto c};d^U_{ab} &\approx R^{z\mapsto a};d^U_{cb},\label{prop:disc-ids2} \\
R^{w\mapsto d};d^U_{ab} &\approx R^{w\mapsto b};d^U_{ad},\label{prop:disc-ids3} \\
d^U_{ab};d^U_{cd} &\simeq d^U_{ad};d^U_{cb}.\label{disc-eq:ddindex}
\end{align}
\end{proposition}
\begin{proof}
We compute~\eqref{prop:disc-sabsorb} by applying equations~\eqref{disc-eq:refl} and~\eqref{disc-eq:rd3}:
$$d^U_{D[a]};S^a \simeq d^U_{D[a]};R^{a\mapsto a} \simeq d^U_{D[a]}.$$
Equality~\eqref{prop:disc-ids0} is derived as follows:
\begin{align*}
\bar d^{uv}_{ab};d^{uv}_{cd} &\approx d^{uv}_{ij};\bar d^{uv}_{ab};R^{i\mapsto c};R^{j\mapsto d}\tag{by~\eqref{disc-eq:rd4}} \\
&\approx S^u;S^v;R^{a\mapsto i};R^{b\mapsto j};R^{i\mapsto c};R^{j\mapsto d}\tag{by~\eqref{disc-eq:ddbar1}} \\
&\approx S^u;S^v;R^{a\mapsto j};R^{b\mapsto d};R^{j\mapsto c}.\tag{by~\eqref{disc-eq:rcomm} and~\eqref{disc-eq:trans}}
\end{align*}
For~\eqref{prop:disc-ids1}, we first use~\eqref{prop:disc-ids0} to get
$$d^{uv}_{cd};\bar d^{uv}_{zw};d^{uv}_{ab}\simeq d^{uv}_{cd};S^u;S^v;R^{z\mapsto a};R^{w\mapsto b},$$
where on the right-hand side we used~\eqref{disc-eq:rcomm} and~\eqref{disc-eq:trans}. Observing that~\eqref{disc-eq:ddbar1} applies on the left-hand side, and simplifying using~\eqref{disc-eq:sr1},~\eqref{disc-eq:sd2} and~\eqref{disc-eq:rd1}, we obtain precisely~\eqref{prop:disc-ids1}. Identities~\eqref{prop:disc-ids2} and~\eqref{prop:disc-ids3} are derived similarly, by letting $w=d$ and $z=c$, respectively.

Identity~\eqref{disc-eq:ddindex} is derived as follows:
\begin{align*}
d^U_{ab};d^U_{cd} &\simeq d^U_{ad};R^{d\mapsto b};d^U_{cd}\tag{by~\eqref{disc-eq:rd3}} \\
&\simeq d^U_{ad};R^{d\mapsto d};d^U_{cb}\tag{by~\eqref{prop:disc-ids3}} \\
&\simeq d^U_{ad};d^U_{cb}.\tag{by~\eqref{disc-eq:rd3}}
\end{align*}
\end{proof}

\subsection{Normal Form}\label{subsec:normal-form}

In this subsection, we define a normal form (Definition~\ref{def:normal-form}), and show that every term is equal to a term in a normal form under the equalities of $\Disc$ (Proposition~\ref{prop:normal-form-existence}). We also identify a class of syntactic manipulations of terms in a normal form (Definition~\ref{def:nf-equivalence}) that both keep the normal form and preserve equality (Lemma~\ref{lma:nf-equivalence}). These results are used in the next section to prove completeness.

\begin{definition}[$ICE$-form]\label{def:ICE-form}
We say that a term is in an {\em $ICE$-form} if it is either an identity term, or if it has the following structure:
$$\mathtt I;\mathtt C;\mathtt E^{<0};\mathtt E^{\geq 0};\bar{\mathtt E}^{\geq 0};\bar{\mathtt E}^{<0};\bar{\mathtt C};\bar{\mathtt I};\mathtt R;\mathtt S,$$
where every letter is a sequence of generating terms of the corresponding kind.
\end{definition}

\begin{proposition}\label{prop:ICE-form}
Any term is equal to a term in an $ICE$-form.
\end{proposition}
\begin{proof}[Proof sketch]
The proof proceeds by repeated inductions: one first shows that all $I$-terms can always be commuted to the left, then that all $C$-terms can be commuted to the left of anything that is not an $I$-term, and so on. We give the full details of the induction in the sequence of Lemmas and Corollaries~\ref{lma:I-commutes}-\ref{cor:ICE-form} in the Appendix~\ref{sec:appendix}.
\end{proof}

\begin{definition}[Renaming form]\label{def:renaming-form}
A well-typed sequence of renaming terms $\mathtt R:H\rightarrow G$ is in a {\em renaming form} if there are sets of vertex names $A=\{a_1,\dots,a_n\}$, $B=\{b_1,\dots,b_n\}$, $C=\{c_1,\dots,c_m\}$ and $D=\{d_1,\dots,d_m\}$ such that
\begin{enumerate}[label=(\arabic*)]
\item $\mathtt R$ can be split into two sequences $\mathtt R = \mathtt A;\mathtt B$ with
$$\mathtt A = R^{a_1\mapsto b_1};\dots;R^{a_n\mapsto b_n}\quad \text{and}\quad \mathtt B = R^{c_1\mapsto d_1};\dots;R^{c_m\mapsto d_m},$$
where $\mathtt B$ can be possibly empty,\label{renaming-form-1}
\item $A\cap B=\eset$,\label{renaming-form-2}
\item $C\sse B$,\label{renaming-form-3}
\item $D\sse A$,\label{renaming-form-4}
\item if $c_i\in C$ and $b_j\in B$ is the unique element such that $b_j=c_i$, then $\Nbr_H(a_j)\neq\Nbr_H(d_i)$.\label{renaming-form-5}
\end{enumerate}
\end{definition}

\begin{lemma}\label{lma:renaming-form}
Any well-typed sequence of renaming terms is equal to a term $\mathtt R;\mathtt S$, where $\mathtt R = \mathtt A;\mathtt B$ is in a renaming form and $\mathtt S$ is a sequence of $S$-terms.
\end{lemma}
\begin{proof}
The idea is that if a vertex $a$ is to be renamed to $b$, then $a\in A$, and we have two cases: (1) $b$ does not already occur in the original chemical graph, and (2) $b$ does occur in the original graph. If (1), then $R^{a\mapsto b}\in\mathtt A$ and $b\in B\setminus C$. If (2), then we first rename $a$ using some ``dummy'' name $c$, so that $R^{a\mapsto c}\in\mathtt A$, $R^{c\mapsto b}\in\mathtt B$, $c\in C$ and $b\in D$. Note that condition~\ref{renaming-form-4} of the renaming form is satisfied, as $b$ must itself be renamed in order for the vertex name become free. Any term of the form $R^{a\mapsto a}$ is replaced by $S^a$.

The formal proof proceeds by induction on the length of the original sequence.

The term $R^{a\mapsto b}$ is equal to $S^a$ if $a=b$, or is already in a renaming form by taking $A=\{a\}$, $B=\{b\}$ and $C=D=\eset$ if $a\neq b$.

Suppose that the statement of the lemma holds for all sequences of renaming terms of length at most $n$. Let $\mathtt R$ be such a sequence of length $n$ such that $\mathtt R;R^{a\mapsto b}$ is defined. By the inductive hypothesis, we may assume that $\mathtt R=\mathtt A;\mathtt B;\mathtt S$ where $\mathtt A;\mathtt B$ is a renaming form with vertex name sets $A$, $B$, $C$ and $D$ as in Definition~\ref{def:renaming-form}. Using the equations for $S$- and $R$-terms, we may commute $R^{a\mapsto b}$ past $\mathtt S$, possibly changing the vertex name $a$, so that it suffices to show that the lemma holds for $\mathtt A;\mathtt B;R^{a\mapsto b}$. If $a=b$, the sequence is equal to $\mathtt A;\mathtt B;S^a$ and we are done; hence assume that $a\neq b$. Note that it follows that $a\notin A\setminus D$ and $a\notin C$, as every vertex name in $A\setminus D$ or $C$ is removed, without being reintroduced. Similarly, we have that $b\notin D$ and $b\notin B\setminus C$. Moreover, if $b\in C$, rename the occurrence of $b$ in both $\mathtt A$ and $\mathtt B$ with a fresh vertex name, updating the sets $C$ and $B$ accordingly. Thus we may assume that $b\notin B$. The remaining cases are as follows.

\noindent\textbf{Case 1:} $a\notin A\cup B$.

\textbf{Subcase 1.1:} $b\notin A$. We rewrite the term to $\mathtt A;R^{a\mapsto b};\mathtt B$ and update the sets $A\mapsto A\cup\{a\}$ and $B\mapsto B\cup\{b\}$.

\textbf{Subcase 1.2:} $b\in A$. It follows that $b\in A\setminus D$, so that $R^{b\mapsto z}\in\mathtt A$ and $a,b$ do not appear in $\mathtt B$. If $\Nbr(a)\neq\Nbr(b)$, let $c$ be a fresh vertex name. We rewrite the term to $\mathtt A;R^{a\mapsto c};R^{c\mapsto b};\mathtt B$ and update the sets $A\mapsto A\cup\{a\}$, $B\mapsto B\cup\{c\}$, $C\mapsto C\cup\{c\}$ and $D\mapsto D\cup\{b\}$. If $\Nbr(a)=\Nbr(b)$, we use equation~\eqref{disc-eq:rsymm} to rewrite $R^{b\mapsto z};R^{a\mapsto b}$ to $S^b;R^{a\mapsto z}$, which reduces the number of $R$-terms to $n$, so that the inductive hypothesis applies.

\noindent\textbf{Case 2:} $a\in A$. It follows that $a\in D$. Now $R^{a\mapsto b}$ commutes with all other terms in $\mathtt B$ except for the unique term $R^{c_i\mapsto d_i}$ where $d_i=a$. But $R^{c_i\mapsto a};R^{a\mapsto b}\equiv R^{c_i\mapsto b}$, which reduces the length of the sequence to $n$, so it is has a renaming form by the inductive hypothesis.

\noindent\textbf{Case 3:} $a\in B$. It follows that $a\in B\setminus C$.

\textbf{Subcase 3.1:} $b\notin A$. Now $R^{a\mapsto b}$ commutes with all the terms in $\mathtt B$, and with all other terms in $\mathtt A$ except for the unique term $R^{a_i\mapsto b_i}$ where $b_i=a$. But $R^{a_i\mapsto a};R^{a\mapsto b}\equiv R^{a_i\mapsto b}$, which reduces the length of the sequence to $n$, so it is has a renaming form by the inductive hypothesis.

\textbf{Subcase 3.2:} $b\in A$. It follows that $b\in A\setminus D$. Now $R^{a\mapsto b}$ commutes with all the terms in $\mathtt B$, and with all other terms in $\mathtt A$ except for the terms $R^{a_i\mapsto a}$ and $R^{b\mapsto b_j}$. There are two options: (1) $a_i=b$ and $b_j=a$, so that these are the same term, (2) the terms are distinct, in which case they commute. In both cases, we use the substitution $R^{a_i\mapsto a};R^{a\mapsto b}\equiv R^{a_i\mapsto b}$ to reduce the length of the sequence, so that the inductive hypothesis applies.

This completes the induction.
\end{proof}

A term is said to be in an $ICER$-form if it is in an $ICE$-form whose sequence of renaming terms is in a renaming form (or is empty).

\begin{definition}[Normal form]\label{def:normal-form}
Let
$$\mathtt t = \mathtt I;\mathtt C;\mathtt E^{<0};\mathtt E^{\geq 0};\bar{\mathtt E}^{\geq 0};\bar{\mathtt E}^{<0};\bar{\mathtt C};\bar{\mathtt I};\mathtt A;\mathtt B;\mathtt S$$
be a term in an $ICER$-form. Let us denote the sets of vertex names in the renaming form by $A_{\mathtt t}$, $B_{\mathtt t}$, $C_{\mathtt t}$ and $D_{\mathtt t}$. Let us additionally define the following sets of vertex names occurring in $\mathtt t$:
\begin{itemize}
\item $D^{add}_{\mathtt t}\coloneqq\left\{a\in\VS : d^U_{D[a]}\in\mathtt t\right\}$ -- the vertex names appearing as subscripts in the disconnections,
\item $D^{remove}_{\mathtt t}\coloneqq\left\{a\in\VS : \bar d^U_{D[a]}\in\mathtt t\right\}$ -- the vertex names appearing as subscripts in the connections,
\item $U_{\mathtt t}\coloneqq\left\{v\in\VS : d^{U[v]}_D\in\mathtt t\text{ or } \bar d^{U[v]}_D\in\mathtt t\right\}$ -- the vertex names appearing as superscripts of the (dis)connections,
\item $S_{\mathtt t}\coloneqq\left\{u\in\VS : S^u\in\mathtt t\right\}$ -- the vertex names appearing in the $S$-terms.
\end{itemize}
We say that a term $\mathtt t$ is in a {\em normal form} if it is in an $ICER$-form as above, and additionally the following conditions hold:
\begin{enumerate}[label=(\arabic*)]
\item for every $u\in S_{\mathtt t}$, the term $S^u$ occurs in $\mathtt t$ exactly once,\label{nf:S0}
\item $\left(U_{\mathtt t}\cup A_{\mathtt t}\cup B_{\mathtt t}\right)\cap S_{\mathtt t}=\eset$,\label{nf:S1}
\item $D^{add}_{\mathtt t}\setminus D^{remove}_{\mathtt t}\sse A_{\mathtt t}\setminus D_{\mathtt t}$,\label{nf:discrename}
\item $D^{add}_{\mathtt t}\cap B_{\mathtt t}=\eset$,\label{nf:dbdisjoint}
\item if a connection $\bar d^U_{D[a]}:A\rightarrow B$ and a renaming term $R^{z\mapsto a}$ both occur, then $A$ is not in the domain of $\bar d^U_{D[z/a]}$,\label{nf:connrename}
\item if $d\neq I$ and a disconnection $d^U_D$ occurs in $\mathtt t$, then the connections $\bar d^U_{F}$ and $\bar d^{U^r}_{F}$ do not occur in $\mathtt t$ for any $F$ (here $U^r$ denotes the reverse string),\label{nf:disconnection}
\item if the disconnection $E^{uv}$ occurs in $\mathtt t$, then for any vertex name $w\in\VS$, the connection $\bar E^{uw}$ does not occur in $\mathtt t$,\label{nf:electron}
\item if the disconnection $I^{uv}$ and the connection $\bar I^{uv}$ both occur in $\mathtt t$, then one of the terms $E^v_D$, $\bar E^v_D$, $E^{va}$ or $\bar E^{va}$ occurs in $\mathtt t$.\label{nf:ion}
\end{enumerate}
\end{definition}

\begin{proposition}\label{prop:normal-form-existence}
In $\Disc$, any term is equal to a term in normal form.
\end{proposition}
\begin{proof}
By Propositions~\ref{prop:ICE-form} and~\ref{lma:renaming-form}, every term is equal to a term in an $ICER$-form: let us fix such a term
$$\mathtt t = \mathtt I;\mathtt C;\mathtt E^{<0};\mathtt E^{\geq 0};\bar{\mathtt E}^{\geq 0};\bar{\mathtt E}^{<0};\bar{\mathtt C};\bar{\mathtt I};\mathtt A;\mathtt B;\mathtt S.$$

Conditions~\ref{nf:S0} and~\ref{nf:S1} are obtained by absorbing the ``excess'' $S$-terms into other terms using equations~\eqref{disc-eq:sr2},~\eqref{disc-eq:sidem},~\eqref{disc-eq:sd2} and~\eqref{prop:disc-sabsorb}. Conditions~\ref{nf:discrename} and~\ref{nf:dbdisjoint} are obtained by treating all the vertex names in $D^{add}_{\mathtt t}$ as ``dummy'' names, which are removed either by a connection or a renaming term.

For~\ref{nf:connrename}, suppose that $\bar d^U_{D[a]}:A\rightarrow B$ and $R^{z\mapsto a}$ both occur, and moreover $A$ is in the domain of $\bar d^U_{D[z/a]}$. We commute the renaming term to the left to obtain $\bar d^U_{D[a]};R^{z\mapsto a}$. But this is equal to $\bar d^U_{D[z/a]};R^{a\mapsto a}$ by equations~\eqref{prop:disc-ids2} and~\eqref{prop:disc-ids3}, which gets rid of the renaming term by $R^{a\mapsto a} \equiv S^a$.

For~\ref{nf:disconnection} and~\ref{nf:electron}, we consider the three cases separately.

\noindent\textbf{Case 1:} $E^{uv}$ and $\bar E^{uw}$ occur in $\mathtt t$. We commute the terms so that they occur one after the other $E^{uv};\bar E^{uw}$, which by equations~\eqref{disc-eq:ddbar4} and~\eqref{disc-eq:eebar} is equal to some combination of $S$- and $R$-terms. Noticing that the rewriting procedure to obtain an ICE-form either commutes or absorbs $S$- and $R$-terms (specifically, the results from Lemma~\ref{lma:Enonnegbar-commutes} onwards apply), we conclude that $\mathtt t$ has an ICE-form without $E^{uv}$ and $\bar E^{uw}$.

\noindent\textbf{Case 2:} $E^u_{ab}$ and $\bar E^u_{cd}$ occur in $\mathtt t$. In combination with Case~1, it follows that the terms $E^{uv}$ and $\bar E^{uw}$ do not occur, so that there are no obstruction for commuting $E^u_{ab}$ next to $\bar E^u_{cd}$, obtaining $E^u_{ab};\bar E^u_{cd}$. By equations~\eqref{disc-eq:ddbar1},~\eqref{disc-eq:ddbar2} and~\eqref{disc-eq:ddbar3}, this is equal to some combination of $S$- and $R$-terms, which we commute to the right as in Case~1.

\noindent\textbf{Case 3:} $C^{uv}_{ab}$ and either $\bar C^{uv}_{cd}$ or $\bar C^{vu}_{cd}$ occur in $\mathtt t$. The latter case simply reduces to the former by equation~\eqref{disc-eq:cs}. Thus suppose that $\bar C^{uv}_{cd}$ occurs. First, we use equation~\eqref{disc-eq:comm9} to commute $C^{uv}_{ab}$ to the right past all the $E^{<0}$-terms, and $\bar C^{uv}_{cd}$ to the left past all the $\bar E^{<0}$-terms. Next, we commute any terms of the form $E^{ui}$ and $E^{vj}$ to the right past the $\bar E^{\geq 0}$-terms and $\bar C^{uv}_{cd}$ using the fact that by Case~1 the terms $\bar E^{ui}$ and $\bar E^{vj}$ do not occur, together with the equations~\eqref{disc-eq:comm2} and~\eqref{disc-eq:comm11}. Now there are no obstructions for commuting $C^{uv}_{ab}$ to the right past all the $E^{\geq 0}$- and $\bar E^{\geq 0}$-terms, obtaining $C^{uv}_{ab};\bar C^{uv}_{cd}$. As in Case~2, equations~\eqref{disc-eq:ddbar1},~\eqref{disc-eq:ddbar2} and~\eqref{disc-eq:ddbar3} yield that this is equal to some combination of $S$- and $R$-terms, which we commute to the right as in Case~1. Finally, we return the terms of the form $E^{ui}$ and $E^{vj}$ back to the left past all the $\bar E^{\geq 0}$-terms.

For~\ref{nf:ion}, suppose that both $I^{uv}$ and $\bar I^{uv}$ occur in $\mathtt t$ such that no $E$- or $\bar E$ term containing $u$ occurs. It follows that no $E$- or $\bar E$ term containing $v$ occurs either: application of $\bar I^{uv}$ requires for $u$ and $v$ to have equal and opposite charge, so if the charge of $u$ is unchanged, so is the charge of $v$; moreover, by~\ref{nf:disconnection} and~\ref{nf:electron}, we may assume that no change can be reversed, so that we indeed cannot have any $E$- or $\bar E$ term containing $v$. But now there are no obstructions for commuting $I^{uv}$ all the way to the right until we obtain $I^{uv};\bar I^{uv}$, which is equal to $S^u;S^v$ by~\eqref{disc-eq:ddbar4}.
\end{proof}

\begin{definition}[Normal form equivalence]\label{def:nf-equivalence}
Let
$$\mathtt t = \mathtt I;\mathtt C;\mathtt E^{<0};\mathtt E^{\geq 0};\bar{\mathtt E}^{\geq 0};\bar{\mathtt E}^{<0};\bar{\mathtt C};\bar{\mathtt I};\mathtt A;\mathtt B;\mathtt S$$
be a term in a normal form. Define the following syntactic manipulations of $\mathtt t$:
\begin{enumerate}
\item commuting the terms inside each of the named sequences in the normal form,\label{nf-eq:comm}
\item permuting vertex names in $C$-terms: if the term $C^{uv}_{ab}$ occurs, we may substitute it with $C^{vu}_{ba}$,\label{nf-eq:permute}
\item if $d\in\{C,E,\bar C,\bar E\}$ such that $d^U_{ab};d^U_{cd}$ occurs, we may substitute $d^U_{ab};d^U_{cd}\mapsto d^U_{ad};d^U_{cb}$,\label{nf-eq:same-swap}
\item renaming of vertices that are introduced and removed: if $a\in D^{add}_{\mathtt t}\cup C_{\mathtt t}$ and $z\in\VS$ does not occur in $\mathtt t$ or its domain, then we may substitute $\mathtt t\mapsto\mathtt t[z/a]$,\label{nf-eq:rename}
\item exchanging vertex names between renaming terms: if both $R^{a\mapsto b}$ and $R^{c\mapsto d}$ occur in $\mathtt A$ such that $\Nbr(a)=\Nbr(c)$, we may swap $a$ and $c$,\label{nf-eq:r-swap}
\item exchanging vertex names between connections and renaming terms: if $d\in\{E,C\}$, and $\bar d^U_{D[a]}:A\rightarrow B$ and $R^{z\mapsto b}$ both occur such that $A$ is in the domain of $\bar d^U_{D[z/a]}$, then we may substitute $\bar d^U_{D[a]}\mapsto \bar d^U_{D[z/a]}$ and $R^{z\mapsto b}\mapsto R^{a\mapsto b}$.\label{nf-eq:r-exchange}
\end{enumerate}
We say that two terms $\mathtt t$ and $\mathtt s$ in a normal form are {\em equivalent}, written $\mathtt t\sim\mathtt s$, if one can be obtained from the other by a sequence of the syntactic manipulations defined above.
\end{definition}
Observing that each syntactic manipulation in Definition~\ref{def:nf-equivalence} is reversible, we see that $\sim$ is an equivalence relation on the set of terms in normal form.
\begin{lemma}\label{lma:nf-equivalence}
Let $\mathtt t$ and $\mathtt s$ be terms in normal forms such that $\mathtt t\sim\mathtt s$. Then $\mathtt t \equiv \mathtt s$.
\end{lemma}
\begin{proof}
This follows by noticing that every syntactic manipulation of Definition~\ref{def:nf-equivalence} keeps the term in a normal form, and moreover preserves the equality $\equiv$:
\begin{enumerate}
\item the terms may be commuted by equations~\eqref{disc-eq:comm1},~\eqref{disc-eq:rcomm} and~\eqref{disc-eq:scomm},
\item vertex names in $C$-terms may be permuted by~\eqref{disc-eq:cs},
\item the indices in repeated (dis)connections may be exchanged by~\eqref{disc-eq:ddindex},
\item if $a\in D^{add}_{\mathtt t}$, so that there is a disconnection $d^U_{D[a]}$, we use equation~\eqref{disc-eq:rd3} to obtain $d^U_{D[a]} \equiv d^U_{D[z/a]};R^{z\mapsto a}$ to introduce the desired fresh variable $z$; the renaming term can then be absorbed into the second occurrence of $a$, hence replacing $a$ with $z$ (the case when $a\in C_{\mathtt t}$ is similar),
\item the renaming of $\alpha$-vertices with the same neighbour can be exchanged by equation~\eqref{disc-eq:rsymm}:
$$R^{a\mapsto b};R^{c\mapsto d} \equiv S^c;R^{a\mapsto b};R^{c\mapsto d} \equiv R^{c\mapsto b};R^{a\mapsto c};R^{c\mapsto d} \equiv R^{c\mapsto b};R^{a\mapsto d},$$
\item the last syntactic manipulation is obtained by equations~\eqref{prop:disc-ids2} and~\eqref{prop:disc-ids3}.
\end{enumerate}
\end{proof}

\section{From Disconnections to Reactions, Functorially}\label{sec:disc-to-react}

We are finally ready to tie together the constructions in the previous sections: here we construct a functor $R:\Disc\rightarrow\React$ and prove that it is faithful and full up to an isomorphism, hence establishing the claimed soundness, completeness and universality results. Proving faithfulness (completeness) relies crucially on the normal form results (Subsection~\ref{subsec:normal-form}). In combination, the results of this section allow for algebraic reasoning about the reactions using the equations for the disconnection rules (Figure~\ref{fig:disc-axioms}).

We define a function $R$ from terms to morphisms in $\React$ as follows. Given a term $\mathtt t:A\rightarrow B$, the morphism $R(\mathtt t):A\rightarrow B$ has the form $R(R_1(\mathtt t),R_2(\mathtt t),\id,\id)$, where $\id:\Chem{R_1(\mathtt t)}\rightarrow\Chem{R_2(\mathtt t)}$ and $\id:V_A\setminus R_1(\mathtt t)\rightarrow V_B\setminus R_2(\mathtt t)$ are both identity maps. Since all the terms are mapped to morphisms whose bijection and isomorphism parts are the identities, we omit these, and write $R(\mathtt t)=(R_1(\mathtt t),R_2(\mathtt t))$. The recursive definition of this mapping is given below:
\begin{align*}
R(\id_A) &\coloneqq (\eset,\eset) & R(E^{uv}) &\coloneqq (\{u,v\},\{u,v\}) \\
R\left(S^u\right) &\coloneqq (\{u\},\{u\}) & R(I^{uv}) &\coloneqq (\{u,v\},\{u,v\}) \\
R\left(R^{u\mapsto v}\right) &\coloneqq (\{u\},\{v\}) & R(C^{uv}_{ab}) &\coloneqq (\{u,v\},\{u,v,a,b\}) \\
R(E^{u}_{ab}) &\coloneqq (\{u\},\{u,a,b\}) & R\left(\bar d^{uv}_{ab}\right) &\coloneqq\overline{R\left(d^{uv}_{ab}\right)} \\
 &\phantom{\coloneqq} & R(\mathtt t;\mathtt s) &\coloneqq R(\mathtt t);R(\mathtt s).
\end{align*}
Observe that for all the disconnections we have $R(d^U_D)=(U,U\cup D)$.

Soundness of disconnection rules with respect to reactions is expressed as functoriality:
\begin{proposition}\label{prop:r-dagger-functor}
The assignment $R:\Disc\rightarrow\React$ is a dagger functor.
\end{proposition}
\begin{proof}
Functoriality and preservation of dagger structure follow immediately from the definition. We have to show that $R$ preserves the equalities in $\Disc$, generated by the identities in Figure~\ref{fig:disc-axioms}. Most of these follow immediately by assuming that the expressions on both sides of the equality have the same type and are identical once $R$ is applied, hence we only give the cases that are less obvious or require more computation. To further simplify the notation, we omit the curly brackets of set-builder notation as well as the commas separating vertex names from each other: so e.g.~$(uv,uvab)$ stands for $(\{u,v\},\{u,v,a,b\})$.

For~\eqref{disc-eq:rd1}, suppose that $R^{u\mapsto v};d^U_D\equiv d^U_D;R^{u\mapsto v}$. In particular, it follows that $u,v\notin U\cup D$. Let us write $R(d^U_D)=(R_1,R_2)$, so that $u,v\notin R_1,R_2$. We use this to show that both sides of the equality evaluate to the same map:
$$(u,v);(R_1,R_2) = (\{u\}\cup R_1, R_2\cup\{v\}) = (R_1,R_2);(u,v).$$

For~\eqref{disc-eq:rd4}, suppose that $d^{U'}_{ij};\bar h^U_{ab};R^{i\mapsto c};R^{j\mapsto d} \equiv \bar h^U_{ab};d^{U'}_{cd}$. Denote $R(d^{U'}_{ij})=(u'v',u'v'ij)$, $R(d^{U'}_{cd})=(u'v',u'v'cd)$ and $R(\bar h^{U}_{ab})=(uvab,uv)$. Note that $d$ and $h$ are not $E^{\geq 0}$-terms, whence it follows that $i,j\notin U$ and $a,b\notin U'$. From the fact that the left-hand side is defined, we obtain that $\{i,j\}$ and $\{a,b\}$ are disjoint. The left-hand side is thus translated to
\begin{flalign*}
(u'v',u'v'ij);(uvab,uv);(i,c);(j,d) &= (u'v'uvab,uvu'v'ij);(i,c);(j,d) \\
                              &= (u'v'uvab,cuvu'v'j);(j,d) \\
                              &= (u'v'uvab,dcuvu'v') \\
                              &= (uvu'v'ab,uvu'v'cd) \\
                              &= (uvab,uv);(u'v',u'v'cd),
\end{flalign*}
which we recognise as the translation of the right-hand side.

For~\eqref{disc-eq:eebar}, suppose $E^{ua};\bar E^{ub} \equiv S^u;R^{a\mapsto z};R^{b\mapsto a};R^{z\mapsto b}$. We start from the translation of the right-hand side:
\begin{flalign*}
(u,u);(a,z);(b,a);(z,b) &= (ua,zu);(bz,ba) \\
                        &= (uab,bau) \\
                        &= (uab,uba) \\
                        &= (ua,ua);(ub,ub),
\end{flalign*}
which we recognise as the translation of the left-hand side.

For~\eqref{disc-eq:comm1}, write $R(d^U_D)=(U,U\cup D)$ and $R(d^{U'}_{D'})=(U',U'\cup D')$, so that we get
$$R(d^U_D;d^{U'}_{D'}) = (U\cup U', U\cup U'\cup D\cup D') = R(d^{U'}_{D'};d^U_D).$$
\end{proof}

Recall the syntactic manipulations of terms in normal form we introduced in Definition~\ref{def:nf-equivalence}. We have seen that these manipulations preserve equality (Lemma~\ref{lma:nf-equivalence}). The following lemma is the core of the completeness argument.
\begin{lemma}\label{lma:nf-equality-equivalence}
Let $\mathtt t$ and $\mathtt s$ be terms in a normal form such that $R(\mathtt t)=R(\mathtt s)$. Then $\mathtt t\sim\mathtt s$.
\end{lemma}
\begin{proof}
Let us write
\begin{align*}
\mathtt t &= \mathtt I_{\mathtt t};\mathtt C_{\mathtt t};\mathtt E^{<0}_{\mathtt t};\mathtt E^{\geq 0}_{\mathtt t};\bar{\mathtt E}^{\geq 0}_{\mathtt t};\bar{\mathtt E}^{<0}_{\mathtt t};\bar{\mathtt C}_{\mathtt t};\bar{\mathtt I}_{\mathtt t};\mathtt A_{\mathtt t};\mathtt B_{\mathtt t};\mathtt S_{\mathtt t}, \\
\mathtt s &= \mathtt I_{\mathtt s};\mathtt C_{\mathtt s};\mathtt E^{<0}_{\mathtt s};\mathtt E^{\geq 0}_{\mathtt s};\bar{\mathtt E}^{\geq 0}_{\mathtt s};\bar{\mathtt E}^{<0}_{\mathtt s};\bar{\mathtt C}_{\mathtt s};\bar{\mathtt I}_{\mathtt s};\mathtt A_{\mathtt s};\mathtt B_{\mathtt s};\mathtt S_{\mathtt s}.
\end{align*}
Similarly, let us denote the vertex name sets in the respective renaming forms by $A_{\mathtt t},B_{\mathtt t},C_{\mathtt t},D_{\mathtt t}$ and $A_{\mathtt s},B_{\mathtt s},C_{\mathtt s},D_{\mathtt s}$. Let us denote the morphism $R(\mathtt t)=R(\mathtt s)$ by $(R_1,R_2):A\rightarrow B$.

First, we observe that if $E^{ua}\in\mathtt E^{\geq 0}_{\mathtt t}$, then condition~\ref{nf:electron} of a normal form implies that the charge of $u$ cannot be increased, whence there is a vertex name $b\in\VS$ such that $E^{ub}\in\mathtt E^{\geq 0}_{\mathtt s}$. Similarly, by condition~\ref{nf:disconnection}, if $E^u_{ab}\in\mathtt E^{<0}_{\mathtt t}$, then $E^u_{cd}\in\mathtt E^{<0}_{\mathtt s}$ for some $c,d\in\VS$; and if $C^{uv}_{ab}\in\mathtt C_{\mathtt t}$, then $C^{uv}_{cd}\in\mathtt C_{\mathtt s}$ for some $c,d\in\VS$. By condition~\ref{nf:ion}, if $I^{uv}\in\mathtt I_{\mathtt t}$ then $I^{uv}\in\mathtt I_{\mathtt s}$. Since the connections cannot be undoing the disconnections, a similar inclusion up to $\alpha$-vertices holds for them. Thus we obtain that the sequences of disconnections and connections must coincide, up to renaming the $\alpha$-vertices.

Next, suppose that $R^{a\mapsto b}\in\mathtt A_{\mathtt t}$, so that $a\in A_{\mathtt t}$ and $b\in B_{\mathtt t}$. There are four cases.

\noindent\textbf{Case 1:} $a\notin D^{add}_{\mathtt t}$ and $b\notin C_{\mathtt t}$. This implies that $a\in R_1$ and $b\in R_2$. Moreover, if $b\in R_1$, then condition~\ref{nf:connrename} yields that $\Nbr(a)\neq\Nbr(b)$. It follows that either $R^{a\mapsto b}\in\mathtt A_{\mathtt s}$, or both $\bar d^U_{D[a]}$ and $R^{z\mapsto b}$ occur in $\mathtt s$ such that $\bar d^U_{D[z/a]}$ is defined. But in the latter case the vertex names $a$ and $z$ may be exchanged by syntactic manipulation~\ref{nf-eq:r-exchange}, so that we may assume $R^{a\mapsto b}\in\mathtt A_{\mathtt s}$.

\noindent\textbf{Case 2:} $a\notin D^{add}_{\mathtt t}$ and $b\in C_{\mathtt t}$. This means that $R^{b\mapsto d}\in\mathtt B_{\mathtt t}$ for some $d\in D_{\mathtt t}$, and for some $x\in\VS$, we have $R^{d\mapsto x}\in\mathtt A_{\mathtt t}$. Condition~\ref{nf:discrename} implies that $d\notin D^{add}_{\mathtt t}$, so that we have $a\in R_1$ and $d\in R_1\cap R_2$. If $x\notin C_{\mathtt t}$, then by Case~1, $R^{d\mapsto x}\in\mathtt A_{\mathtt s}$, so that also $R^{a\mapsto z}\in\mathtt A_{\mathtt s}$ and $R^{z\mapsto d}\in\mathtt B_{\mathtt s}$ for some $z\in\VS$. If $x\in C_{\mathtt t}$, then we inductively repeat Case~2. By syntactic manipulation~\ref{nf-eq:rename}, we may assume that $R^{a\mapsto b}\in\mathtt A_{\mathtt s}$ and $R^{b\mapsto d}\in\mathtt B_{\mathtt s}$.

\noindent\textbf{Case 3:} $a\in D^{add}_{\mathtt t}$ and $b\notin C_{\mathtt t}$. Thus there is a disconnection $d^U_{D[a]}\in\mathtt t$, so that $d^U_{D[x/a]}\in\mathtt s$ for some $x\in\VS$. This implies $a\notin R_1\cup R_2$ and $b\in R_2$. As in Case~1, it follows that $R^{z\mapsto b}\in\mathtt A_{\mathtt s}$ for some $z\in\VS$. Moreover, in this case we must have $\Nbr(x)=\Nbr(z)$, whence by syntactic manipulation~\ref{nf-eq:r-swap} we may assume that $R^{x\mapsto b}\in\mathtt A_{\mathtt s}$ with $x\in D^{add}_{\mathtt s}$. By syntactic manipulation~\ref{nf-eq:rename}, we may assume that $R^{a\mapsto b}\in\mathtt A_{\mathtt s}$ and $d^U_{D[a]}\in\mathtt s$.

\noindent\textbf{Case 4:} $a\in D^{add}_{\mathtt t}$ and $b\in C_{\mathtt t}$. We thus have $a,b\notin R_1\cup R_2$. This means that $R^{b\mapsto d}\in\mathtt B_{\mathtt t}$ for some $d\in D_{\mathtt t}$, so that $R^{d\mapsto x}\in\mathtt A_{\mathtt t}$ for some $x\in\VS$. Condition~\ref{nf:discrename} implies that $d\notin D^{add}_{\mathtt t}$, so that $d\in R_1\cap R_2$ and either Case~1 or Case~2 applies to $R^{d\mapsto x}$. In both cases we conclude that $R^{d\mapsto x}\in\mathtt A_{\mathtt s}$. Since $d\in R_2$, we have $R^{w\mapsto d}\in\mathtt B_{\mathtt s}$ for some $w\in\VS$. Since there is a disconnection $d^U_{D[a]}\in\mathtt t$, we have $d^U_{D[y/a]}\in\mathtt s$ for some $y\in\VS$. Note that we have $\Nbr(y)=\Nbr(w)$. Consequently, by syntactic manipulation~\ref{nf-eq:r-swap}, we may assume that $R^{y\mapsto w}\in\mathtt A_{\mathtt s}$ with $y\in D^{add}_{\mathtt s}$ and $w\in C_{\mathtt s}$. By syntactic manipulation~\ref{nf-eq:rename}, we may assume that $R^{a\mapsto b}\in\mathtt A_{\mathtt s}$, $R^{b\mapsto d}\in\mathtt B_{\mathtt s}$ and $d^U_{D[a]}\in\mathtt s$.

Thus we have shown that $\mathtt t$ and $\mathtt s$ have the same renaming sequences (up to $\sim$), and up to the syntactic manipulations, $D^{add}_{\mathtt t}=D^{add}_{\mathtt s}$ and $D^{remove}_{\mathtt t}=D^{remove}_{\mathtt s}$.

If $S^u\in\mathtt S_{\mathtt t}$, then $u\in R_1\cap R_2$ and, by conditions~\ref{nf:S0} and~\ref{nf:S1}, $u$ does not occur anywhere else in $\mathtt t$. The argument so far entails that $u\notin U_{\mathtt s}\cup A_{\mathtt s}\cup B_{\mathtt s}$, so that $S^u\in\mathtt S_{\mathtt s}$. Thus $\mathtt S_{\mathtt t} =\mathtt S_{\mathtt s}$.

Now the only difference left between $\mathtt t$ and $\mathtt s$ is in which order the vertex names are introduced and removed. This is taken care of precisely by syntactic manipulations~\ref{nf-eq:comm} and~\ref{nf-eq:same-swap}.
\end{proof}

Combining the above lemma with the results from the previous section, we conclude that the functor $R:\Disc\rightarrow\React$ is faithful. We spell this out in detail in the following:
\begin{theorem}[Completeness]\label{thm:completeness}
For all terms $\mathtt t$ and $\mathtt s$, we have $\mathtt t\equiv\mathtt s$ in $\Disc$ if and only if $R(\mathtt t)=R(\mathtt s)$ in $\React$.
\end{theorem}
\begin{proof}
The `only if' direction is functoriality (Proposition~\ref{prop:r-dagger-functor}). The `if' direction follows from the fact that every term is equal to a term in normal form (Proposition~\ref{prop:normal-form-existence}) and from Lemmas~\ref{lma:nf-equality-equivalence} and~\ref{lma:nf-equivalence}.
\end{proof}

The argument for universality turns out to be much simpler than that for completeness. However, in combination with Theorem~\ref{thm:completeness}, it gives a rather strong representation result for reactions: not only can every reaction be decomposed into a sequence of disconnection rules, but this sequence is also unique, up to changing the vertex names and up to the equations in $\Disc$. In abstract terms, the statement of universality is that the functor $R:\Disc\rightarrow\React$ is full up to isomorphism in $\React$. As for completeness, we spell out the details:
\begin{theorem}[Universality]\label{thm:universality}
Given a reaction $r:A\rightarrow C$ in $\React$, there is a term $\mathtt t:A\rightarrow B$ in $\Disc$ and an isomorphism $\iota:B\xrightarrow{\sim} C$ in $\React$ such that $R(\mathtt t);\iota = r$.
\end{theorem}
\begin{proof}
Observe that every reaction $r:A\rightarrow C$ factorises as
$$(U_A,U_B,\id,\id);(\eset,\eset,!,\iota),$$
where $(U_A,U_B):A\rightarrow B$ is some reaction and $\iota:B\rightarrow C$ is an isomorphism of labelled graphs. Now, we may disconnect all possible bonds inside $U_A$, and then connect all possible bonds to obtain $U_B$. The fact that $U_A$ and $U_B$ have the same atom vertices and the same net charge guarantee that this can always be done. Precisely, the sought-after term $\mathtt t:A\rightarrow B$ is then given by
\begin{flalign*}
&\prod_{\substack{u\in\Crgp{U_A} \\ v\in\Crgn{U_A}}}\left(I^{uv}\right)^{\ion(m_A(u,v))}; \prod_{u,v\in\Chem{U_A}}\prod_{i=1}^{\cov(m_A(u,v))}C^{uv}_{a_ib_i}; \\
&\prod_{u\in\Crgn{U_A}}\prod_{i=1}^{-\tau^{\crg}_A(u)}E^u_{a_ib_i}; \prod_{u\in\Chem{U_A}}\prod_{i=1}^{\mathbf v\tau^{\At}_A(u) - \max\left(\tau^{\crg}_A(u),0\right)}E^{ua_i}; \\
&\prod_{u\in\Chem{U_B}}\prod_{i=1}^{\mathbf v\tau^{\At}_B(u) - \max\left(\tau^{\crg}_B(u),0\right)}\bar E^{ua_i}; \prod_{u\in\Crgn{U_B}}\prod_{i=1}^{-\tau^{\crg}_B(u)}\bar E^u_{a_ib_i}; \\
&\prod_{u,v\in\Chem{U_B}}\prod_{i=1}^{\cov(m_B(u,v))}\bar C^{uv}_{a_ib_i}; \prod_{\substack{u\in\Crgp{U_B} \\ v\in\Crgn{U_B}}}\left(\bar I^{uv}\right)^{\ion(m_B(u,v))}; \\
&\prod_{a\in\alpha(U_A)\setminus D}R^{a\mapsto b_a}; \prod_{b\in\alpha(U_B)}R^{a_b\mapsto b}; \prod_{u\in U_B}S^u,
\end{flalign*}
where the vertex names introduced by the $C$- and $E^{<0}$-terms are chosen to so that they do not appear anywhere in $A$ or $B$, and their set is denoted by $I$. The vertex names removed by the $\bar C$- and $\bar E^{<}$-terms are chosen from $U_A\cup I$ such that the connection is well-typed: their set is denoted by $D$. Similarly, the $\alpha$-vertices appearing in the $E^{\geq 0}$- and $\bar E^{\geq 0}$-terms are chosen from $U_A\cup I$ such that the terms are well-typed. The vertex names introduced by the $R^{a\mapsto b_a}$-terms, where $a\in \alpha(U_A)\setminus D$, are chosen so that they do not appear in $A$, $B$ or $I$: their set is denoted by $R$. Finally, the vertex names removed by the $R^{a_b\mapsto b}$-terms, where $b\in\alpha(U_B)$ are chosen from $I\cup R$ in such a way that the terms are well-typed.

Note that while the term we obtain is in an $ICE$-form, it will not, in general, be in normal form.
\end{proof}

\begin{example}\label{ex:disconnection-sequence}
Following the procedure of Theorem~\ref{thm:universality}, the reaction in Example~\ref{ex:reaction} decomposes into the following sequence of (dis)connection rules:
\begin{align*}
& C^{zu}_{ab};C^{vw}_{cd};C^{ru}_{ij};C^{ru}_{nm};E^{vc};E^{wd};E^{za};E^{ub};E^{ri};E^{uj};E^{rn};E^{um}; \\
& \bar E^{vc};\bar E^{wd};\bar E^{za};\bar E^{ub};\bar E^{ri};\bar E^{uj};\bar E^{rn};\bar E^{um};\bar C^{ru}_{ij};\bar C^{ru}_{nm};\bar C^{wz}_{da};\bar C^{uv}_{bc}; \\
& S^z;S^u;S^v;S^w;S^r.
\end{align*}
The normal form of the above sequence is given by:
$$C^{zu}_{ab};C^{vw}_{cd};\bar C^{wz}_{da};\bar C^{uv}_{bc};S^r.$$
\end{example}

\section{Conclusion}\label{sec:conclusion}

We have formalised and axiomatised the retrosynthetic rules as a category, showing that there is a sound, complete and universal translation into the category of reactions -- a more familiar object to study in computational chemistry.

Universality can be thought of as a consistency result for reactions: their definition captures exactly those rearrangements of chemical graphs which result from local, chemically motivated rewrite rules. Completeness says that there is no redundancy in the representation: treating the (dis)connection rules as terms, the terms can be endowed with equations such that the terms describing the same reaction are identified. As the decomposition of a reaction into a sequence of (dis)connection rules is algorithmic, these results can be used to automatically break a reaction (or its part) into smaller components: the purpose can be, {\em inter alia}, retrosynthetic analysis or storing reaction data in a systematic way.

\subsection{Future Work}

\paragraph{Chemical Questions.} An important part of chemical data is stereochemistry, that is, spatial orientation of the molecule: many molecules of interest (like pharmaceuticals) possess chiral enantiomers (i.e.~molecules that have the same atoms and connectivity, but are mirror images of each other due to spatial orientation) which have different properties. We therefore wish to incorporate spatial data into our categorical representation. While stereochemistry is relatively straightforward to account for on the level of chemical graphs and reactions~\cite{approach-ictac}, it is unclear how to do this for the disconnection rules, as they only operate at one or two vertices at a time. A more straightforward extension of the formalism presented here would introduce energy and dynamics into the disconnection rules by quantifying how much energy each (dis)connection (in a particular context) requires to occur.

\paragraph{Computational Questions.} Given the algorithmic nature of both completeness and universality proofs, the next step is to implement both. The first algorithm would take an arbitrary reaction as an input, and output a sequence of disconnection rules representing it. The second algorithm would decide whether two terms are equal or not, implementing the normalisation procedure. Another direction for connecting this work with more standard approaches to computational chemistry would be translating our formalism to a widely used notation such as SMILES~\cite{smiles88,daylight-smiles}.

\paragraph{Mathematical Questions.} An important mathematical development is to introduce monoidal terms in the disconnection category, so as to allow parallel reactions, as well as usage of graphical calculi for monoidal categories~\cite{selinger,piedeleu-zanasi}. Another mathematical question is whether the categories $\Disc$ and $\React$ have any interesting categorical structure, apart from being dagger categories. Finally, we would like to make precise the connection between the category of reactions and double pushout rewriting~\cite{inferring-rule-composition,intermediate-level,rewriting-life21}.

\bibliographystyle{plain}
\bibliography{bibliography}

\appendix
\section{$ICE$-form}\label{sec:appendix}

Here we give the detailed inductive proof of the fact that any term has an equivalent term in an $ICE$-form (Proposition~\ref{prop:ICE-form}).

\begin{lemma}\label{lma:I-commutes}
Let $\mathtt t$ be a term such that the term $\mathtt t;I^{uv}$ is defined. Then there exists a term $\mathtt t'$ such that $\mathtt t$ and $\mathtt t'$ have the same number of $I$-terms, and one of the following holds:
\begin{enumerate}[label={(\arabic*)}]
\item $\mathtt t;I^{uv}\equiv \mathtt t'$, or
\item there is a disconnection $I^{ab}$ such that $\mathtt t;I^{uv}\equiv I^{ab};\mathtt t'$.
\end{enumerate}
\end{lemma}
\begin{proof}
We proceed by induction on the structure of $\mathtt t$.
\paragraph{Base cases:}
\begin{align*}
\id;I^{uv} &\equiv I^{uv};\id, \\
S^w;I^{uv} &\equiv I^{uv};S^w, \tag{by~\eqref{disc-eq:sd1}} \\
R^{w\mapsto z};I^{uv} &\equiv I^{uv};R^{w\mapsto z}, \tag{by~\eqref{disc-eq:rd1}} \\
I^{wz};I^{uv} &\phantom{\equiv} \text{ is already in the right form} \\
C^{wz}_{ab};I^{uv} &\equiv I^{uv};C^{wz}_{ab}, \tag{by~\eqref{disc-eq:comm3}} \\
E^{w}_{ab};I^{uv} &\equiv I^{uv};E^{w}_{ab}, \tag{by~\eqref{disc-eq:comm4}} \\
E^{wz};I^{uv} &\equiv  I^{uv};E^{wz}, \tag{by~\eqref{disc-eq:comm5}} \\
\bar I^{wz};I^{uv} &\equiv \begin{cases} S^u;S^v \text{ if } w=u, z=v \\
                                    I^{uv};\bar I^{wz} \text{ otherwise,}
                      \end{cases} \tag{by~\eqref{disc-eq:ddbar4-2} and~\eqref{disc-eq:comm2}} \\
\bar C^{wz}_{ab};I^{uv} &\equiv I^{uv};\bar C^{wz}_{ab}, \tag{by~\eqref{disc-eq:comm8}} \\
\bar E^{w}_{ab};I^{uv} &\equiv I^{uv};\bar E^{w}_{ab}, \tag{by~\eqref{disc-eq:comm7}} \\
\bar E^{wz};I^{uv} &\equiv I^{uv};\bar E^{wz}. \tag{by~\eqref{disc-eq:comm6}}
\end{align*}
\paragraph{Inductive case:} Let $\mathtt t:A\rightarrow B$ and $\mathtt s:B\rightarrow C$ be terms such that the statement of the lemma holds. Suppose that the term $\mathtt t;\mathtt s;I^{uv}$ is defined. Then also the term $\mathtt s;I^{uv}$ is defined, so by the inductive hypothesis for $\mathtt s$, there is a term $\mathtt s'$ with the same number of $I$-terms as $\mathtt s$ such that either (1) $\mathtt s;I^{uv}\equiv\mathtt s'$, or (2) $\mathtt s;I^{uv}\equiv I^{ab};\mathtt s'$ for some $I$-term $I^{ab}$. In the first case, we have
$$\mathtt t;\mathtt s;I^{uv} \equiv \mathtt t;\mathtt s',$$
and since $\mathtt t;\mathtt s'$ has the same number of $I$-terms as $\mathtt t;\mathtt s$, it is the sought-after term for the inductive case satisfying (1). In the second case, we have
$$\mathtt t;\mathtt s;I^{uv} \equiv \mathtt t;I^{ab};\mathtt s',$$
so that $\mathtt t;I^{ab}$ is defined. By the inductive hypothesis for $\mathtt t$, there is a term $\mathtt t'$ with the same number of $I$-terms as $\mathtt t$ such that either (1) $\mathtt t;I^{ab}\equiv\mathtt t'$, or (2) $\mathtt t;I^{ab}\equiv I^{wz};\mathtt t'$ for some $I$-term $I^{wz}$. In the first case, we get
$$\mathtt t;\mathtt s;I^{uv} \equiv \mathtt t';\mathtt s',$$
satisfying (1) for the inductive case, as $\mathtt t';\mathtt s'$ and $\mathtt t;\mathtt s$ have the same number of $I$-terms. In the second case, we obtain
$$\mathtt t;\mathtt s;I^{uv} \equiv I^{wz};\mathtt t';\mathtt s',$$
satisfying (2) for the inductive case.
\end{proof}

\begin{corollary}\label{cor:I-form}
Any term is equal to a term of the form $\mathtt I;\mathtt t$, where $\mathtt I$ is a sequence of $I$-terms, and the term $\mathtt t$ contains no $I$-terms.
\end{corollary}

\begin{lemma}\label{lma:C-commutes}
Let $\mathtt t$ be a term not containing any $I$-terms such that the term $\mathtt t;C^{uv}_{ab}$ is defined. Then there exists a term $\mathtt t'$ not containing any $I$-terms such that $\mathtt t$ and $\mathtt t'$ have the same number of $C$-terms, and one of the following holds:
\begin{enumerate}[label={(\arabic*)}]
\item $\mathtt t;C^{uv}_{ab}\equiv\mathtt t'$, or
\item there is a disconnection $C^{wz}_{cd}$ such that $\mathtt t;C^{uv}_{ab}\equiv C^{wz}_{cd};\mathtt t'$.
\end{enumerate}
\end{lemma}
\begin{proof}
By induction on $\mathtt t$.
\paragraph{Base cases:}
\begin{align*}
\id;C^{uv}_{ab} &\equiv C^{uv}_{ab};\id, \\
S^w;C^{uv}_{ab} &\equiv C^{uv}_{ab};S^w, \tag{by~\eqref{disc-eq:sd1}} \\
R^{w\mapsto z};C^{uv}_{ab} &\equiv \begin{cases} C^{uv}_{ab};R^{w\mapsto z} \text{ if } w\neq\{a,b\}, \\
                                            C^{uv}_{kb};R^{a\mapsto z};R^{k\mapsto a} \text{ if } w=a, \\
                                            C^{uv}_{ak};R^{b\mapsto z};R^{k\mapsto b} \text{ if } w=b,
                              \end{cases} \tag{by~\eqref{disc-eq:rd1} and~\eqref{disc-eq:rd3}} \\
C^{wz}_{cd};C^{uv}_{ab} &\phantom{\equiv} \text{ is already in the right form,} \\
E^{w}_{cd};C^{uv}_{ab} &\equiv C^{uv}_{ab};E^{w}_{cd}, \tag{by~\eqref{disc-eq:comm9}} \\
E^{wz};C^{uv}_{ab} &\equiv C^{uv}_{ab};E^{wz}, \tag{by~\eqref{disc-eq:comm10}} \\
\bar I^{wz};C^{uv}_{ab} &\equiv C^{uv}_{ab};\bar I^{wz}, \tag{by~\eqref{disc-eq:comm8}} \\
\bar C^{wz}_{cd};C^{uv}_{ab} &\equiv \begin{cases} S^w;S^z;R^{c\mapsto j};R^{d\mapsto b};R^{j\mapsto a} \text{ if } w=u, z=v, \\
                                                   S^w;S^z;R^{c\mapsto j};R^{d\mapsto a};R^{j\mapsto b} \text{ if } w=v, z=u, \\
                                                   C^{uv}_{ij};\bar C^{wz}_{cd};R^{i\mapsto a};R^{j\mapsto b} \text{ otherwise,}
                                     \end{cases} \tag{by~\eqref{prop:disc-ids0} and~\eqref{disc-eq:rd4}} \\
\bar E^{w}_{cd};C^{uv}_{ab} &\equiv C^{uv}_{ij};\bar E^{w}_{cd};R^{i\mapsto a};R^{j\mapsto b}, \tag{by~\eqref{disc-eq:rd4}} \\
\bar E^{wz};C^{uv}_{ab} &\equiv C^{uv}_{ab};\bar E^{wz}. \tag{by~\eqref{disc-eq:comm11}}
\end{align*}
The inductive case is very similar to that of Lemma~\ref{lma:I-commutes}.
\end{proof}

\begin{corollary}\label{cor:C-form}
Any term is equal to a term of the form $\mathtt I;\mathtt C;\mathtt t$, where $\mathtt I$ and $\mathtt C$ are sequences of $I$-terms and $C$-terms, and the term $\mathtt t$ contains no $I$-terms or $C$-terms.
\end{corollary}

\begin{lemma}\label{lma:Eneg-commutes}
Let $\mathtt t$ be a term not containing any $I$- or $C$-terms such that the term $\mathtt t;E^{u}_{ab}$ is defined. Then there exists a term $\mathtt t'$ not containing any $I$- or $C$-terms such that $\mathtt t$ and $\mathtt t'$ have the same number of $E^{<0}$-terms, and one of the following holds:
\begin{enumerate}[label={(\arabic*)}]
\item $\mathtt t;E^{u}_{ab}\equiv\mathtt t'$, or
\item there is a disconnection $E^{w}_{cd}$ such that $\mathtt t;E^{u}_{ab}\equiv E^{w}_{cd};\mathtt t'$.
\end{enumerate}
\end{lemma}
\begin{proof}
By induction on $\mathtt t$.
\paragraph{Base cases:}
\begin{align*}
\id;E^{u}_{ab} &\equiv E^{u}_{ab};\id, \\
S^w;E^{u}_{ab} &\equiv E^{u}_{ab};S^w, \tag{by~\eqref{disc-eq:sd1}} \\
R^{w\mapsto z};E^{u}_{ab} &\equiv \begin{cases} E^{u}_{ab};R^{w\mapsto z} \text{ if } w\neq\{a,b\}, \\
                                           E^{u}_{kb};R^{a\mapsto z};R^{k\mapsto a} \text{ if } w=a, \\
                                           E^{u}_{ak};R^{b\mapsto z};R^{k\mapsto b} \text{ if } w=b,
                             \end{cases} \tag{by~\eqref{disc-eq:rd1} and~\eqref{disc-eq:rd3}} \\
E^{w}_{cd};E^{u}_{ab} &\phantom{\equiv} \text{ is already in the right form,} \\
E^{wz};E^{u}_{ab} &\equiv E^{u}_{ab};E^{wz}, \tag{by~\eqref{disc-eq:comm12}} \\
\bar I^{wz};E^{u}_{ab} &\equiv E^{u}_{ab};\bar I^{wz}, \tag{by~\eqref{disc-eq:comm7}} \\
\bar C^{wz}_{cd};E^{u}_{ab} &\equiv E^{u}_{ij};\bar C^{wz}_{cd};R^{i\mapsto a};R^{j\mapsto b}, \tag{by~\eqref{disc-eq:rd4}} \\
\bar E^{w}_{cd};E^{u}_{ab} &\equiv \begin{cases} S^w;R^{c\mapsto j};R^{d\mapsto b};R^{j\mapsto a} \text{ if } w=u, \\
                                                 E^u_{ij};\bar E^u_{cd};R^{i\mapsto a};R^{j\mapsto b} \text{ otherwise},
                                   \end{cases} \tag{by~\eqref{prop:disc-ids0} and~\eqref{disc-eq:rd4}} \\
\bar E^{wz};E^u_{ab} &\equiv E^u_{ab};\bar E^{wz}. \tag{by~\eqref{disc-eq:comm13}}
\end{align*}
The inductive case is very similar to that of Lemma~\ref{lma:I-commutes}.
\end{proof}

\begin{corollary}\label{cor:Eneg-form}
Any term is equal to a term of the form $\mathtt I;\mathtt C;\mathtt E^{<0};\mathtt t$, where $\mathtt I$, $\mathtt C$ and $\mathtt E^{<0}$ are sequences of $I$-, $C$-, and $E^{<0}$-terms, and the term $\mathtt t$ contains no $I$-, $C$-, or $E^{<0}$-terms.
\end{corollary}

\begin{lemma}\label{lma:Enonneg-commutes}
Let $\mathtt t$ be a term not containing any $I$-, $C$-, or $E^{<0}$-terms such that the term $\mathtt t;E^{uv}$ is defined. Then there exists a term $\mathtt t'$ not containing any $I$-, $C$-, or $E^{<0}$-terms such that $\mathtt t$ and $\mathtt t'$ have the same number of $E^{\geq 0}$-terms, and one of the following holds:
\begin{enumerate}[label={(\arabic*)}]
\item $\mathtt t;E^{uv}\equiv\mathtt t'$, or
\item there is a disconnection $E^{wz}$ such that $\mathtt t;E^{uv}\equiv E^{wz};\mathtt t'$.
\end{enumerate}
\end{lemma}
\begin{proof}
By induction on $\mathtt t$.
\paragraph{Base cases:}
\begin{align*}
\id;E^{uv} &\equiv E^{uv};\id, \\
S^w;E^{uv} &\equiv E^{uv};S^w, \tag{by~\eqref{disc-eq:sd1}} \\
R^{w\mapsto z};E^{uv} &\equiv \begin{cases} E^{uv};R^{w\mapsto z} \text{ if } z\neq v, \\
                                       E^{uw};R^{w\mapsto v} \text{ if } z=v,
                         \end{cases} \tag{by~\eqref{disc-eq:rd1} and~\eqref{disc-eq:rd2}} \\
E^{wz};E^{uv} &\phantom{\equiv} \text{ is already in the right form,} \\
\bar I^{wz};E^{uv} &\equiv E^{uv};\bar I^{wz}, \tag{by~\eqref{disc-eq:comm6}} \\
\bar C^{wz}_{cd};E^{uv} &\equiv E^{uv};\bar C^{wz}_{cd}, \tag{by~\eqref{disc-eq:comm11}} \\
\bar E^{w}_{cd};E^{uv} &\equiv E^{uv};\bar E^{w}_{cd}, \tag{by~\eqref{disc-eq:comm13}} \\
\bar E^{wz};E^{uv} &\equiv \begin{cases} S^u;S^v \text{ if } w=u \text{ and } z=v, \\
                                    E^{uv};\bar E^{wz} \text{ otherwise.}
                      \end{cases} \tag{by~\eqref{disc-eq:ddbar4-2} and~\eqref{disc-eq:comm2}}
\end{align*}
The inductive case is very similar to that of Lemma~\ref{lma:I-commutes}.
\end{proof}

\begin{corollary}\label{cor:Enonneg-form}
Any term is equal to a term of the form $\mathtt I;\mathtt C;\mathtt E^{<0};\mathtt E^{\geq 0};\mathtt t$, where $\mathtt I$, $\mathtt C$, $\mathtt E^{<0}$ and $\mathtt E^{\geq 0}$ are sequences of $I$-, $C$-, $E^{<0}$, and $E^{\geq 0}$-terms, and the term $\mathtt t$ contains no $I$-, $C$-, $E^{<0}$, or $E^{\geq 0}$-terms.
\end{corollary}

\begin{lemma}\label{lma:Enonnegbar-commutes}
Let $\mathtt t$ be a term not containing any $I$-, $C$- or $E$-terms such that the term $\mathtt t;\bar E^{uv}$ is defined. Then there exists a term $\mathtt t'$ not containing any $I$-, $C$- or $E$-terms such that $\mathtt t$ and $\mathtt t'$ have the same number of $\bar E^{\geq}$-terms, and one of the following holds:
\begin{enumerate}[label={(\arabic*)}]
\item $\mathtt t;\bar E^{uv}\equiv\mathtt t'$, or
\item there is a connection $\bar E^{wz}$ such that $\mathtt t;\bar E^{uv}\equiv\bar E^{wz};\mathtt t'$.
\end{enumerate}
\end{lemma}
\begin{proof}
By induction on $\mathtt t$.
\paragraph{Base cases:}
\begin{align*}
\id;\bar E^{uv} &\equiv\bar E^{uv};\id, \\
S^w;\bar E^{uv} &\equiv \bar E^{uv};S^w, \tag{by~\eqref{disc-eq:sd1}} \\
R^{w\mapsto z};\bar E^{uv} &\equiv \begin{cases} \bar E^{uv};R^{w\mapsto z} \text{ if } z\neq v, \\
                                            \bar E^{uw};R^{w\mapsto v} \text{ if } z=v,
                              \end{cases} \tag{by~\eqref{disc-eq:rd1} and~\eqref{disc-eq:rd2}} \\
\bar I^{wz};\bar E^{uv} &\equiv \bar E^{uv};\bar I^{wz}, \tag{by~\eqref{disc-eq:comm5}} \\
\bar C^{wz}_{cd};\bar E^{uv} &\equiv \bar E^{uv};\bar C^{wz}_{cd}, \tag{by~\eqref{disc-eq:comm10}} \\
\bar E^{w}_{cd};\bar E^{uv} &\equiv \bar E^{uv};\bar E^{w}_{cd}, \tag{by~\eqref{disc-eq:comm12}} \\
\bar E^{wz};\bar E^{uv} &\phantom{\equiv} \text{ is already in the right form.}
\end{align*}
The inductive case is very similar to that of Lemma~\ref{lma:I-commutes}.
\end{proof}

\begin{corollary}\label{cor:Enonnegbar-form}
Any term is equal to a term of the form $\mathtt I;\mathtt C;\mathtt E^{<0};\mathtt E^{\geq 0};\bar{\mathtt E}^{\geq 0};\mathtt t$, where the term $\mathtt t$ contains no $I$-, $C$-, $E$-, or $\bar E^{\geq 0}$-terms.
\end{corollary}

\begin{lemma}\label{lma:Enegbar-commutes}
Let $\mathtt t$ be a term not containing any $I$-, $C$-, $E$-, or $\bar E^{\geq 0}$-terms such that the term $\mathtt t;\bar E^{u}_{ab}$ is defined. Then there exists a term $\mathtt t'$ not containing any $I$-, $C$-, $E$-, or $\bar E^{\geq 0}$-terms such that $\mathtt t$ and $\mathtt t'$ have the same number of $\bar E^{<0}$-terms, and one of the following holds:
\begin{enumerate}[label={(\arabic*)}]
\item $\mathtt t;\bar E^{u}_{ab}\equiv\mathtt t'$, or
\item there is a connection $\bar E^{w}_{cd}$ such that $\mathtt t;\bar E^{u}_{ab}\equiv\bar E^{w}_{cd};\mathtt t'$.
\end{enumerate}
\end{lemma}
\begin{proof}
By induction on $\mathtt t$.
\paragraph{Base cases:}
\begin{align*}
\id;\bar E^{u}_{ab} &\equiv\bar E^{u}_{ab};\id, \\
S^w;\bar E^{u}_{ab} &\equiv \bar E^{u}_{ab};S^w, \tag{by~\eqref{disc-eq:sd1}} \\
R^{w\mapsto z};\bar E^{u}_{ab} &\equiv \begin{cases} \bar E^{u}_{ab};R^{w\mapsto z} \text{ if } z\notin\{a,b\}, \\
                                                \bar E^{u}_{wb} \text{ if } z=a, \\
                                                \bar E^{u}_{aw} \text{ if } z=b,
                                  \end{cases} \tag{by~\eqref{disc-eq:rd1} and~\eqref{disc-eq:rd3}} \\
\bar I^{wz};\bar E^{u}_{ab} &\equiv \bar E^{u}_{ab};\bar I^{wz}, \tag{by~\eqref{disc-eq:comm4}} \\
\bar C^{wz}_{cd};\bar E^{u}_{ab} &\equiv \bar E^{u}_{ab};\bar C^{wz}_{cd}, \tag{by~\eqref{disc-eq:comm9}} \\
\bar E^{w}_{cd};\bar E^{u}_{ab} &\phantom{\equiv} \text{ is already in the right form.}
\end{align*}
The inductive case is very similar to that of Lemma~\ref{lma:I-commutes}.
\end{proof}

\begin{corollary}\label{cor:Enegbar-form}
Any term is equal to a term of the form $\mathtt I;\mathtt C;\mathtt E^{<0};\mathtt E^{\geq 0};\bar{\mathtt E}^{\geq 0};\bar{\mathtt E}^{<0};\mathtt t$, where the term $\mathtt t$ contains no $I$-, $C$-, $E$- or $\bar E$-terms.
\end{corollary}

\begin{lemma}\label{lma:Cbar-commutes}
Let $\mathtt t$ be a term not containing any $I$-, $C$-, $E$- or $\bar E$-terms such that the term $\mathtt t;\bar C^{uv}_{ab}$ is defined. Then there exists a term $\mathtt t'$ not containing any $I$-, $C$-, $E$- or $\bar E$-terms such that $\mathtt t$ and $\mathtt t'$ have the same number of $\bar C$-terms, and one of the following holds:
\begin{enumerate}[label={(\arabic*)}]
\item $\mathtt t;\bar C^{uv}_{ab}\equiv\mathtt t'$, or
\item there is a connection $\bar C^{wz}_{cd}$ such that $\mathtt t;\bar C^{uv}_{ab}\equiv\bar C^{wz}_{cd};\mathtt t'$.
\end{enumerate}
\end{lemma}
\begin{proof}
By induction on $\mathtt t$.
\paragraph{Base cases:}
\begin{align*}
\id;\bar C^{uv}_{ab} &\equiv \bar C^{uv}_{ab};\id, \\
S^w;\bar C^{uv}_{ab} &\equiv \bar C^{uv}_{ab};S^w, \tag{by~\eqref{disc-eq:sd1}} \\
R^{w\mapsto z};\bar C^{uv}_{ab} &\equiv \begin{cases} \bar C^{uv}_{ab};R^{w\mapsto z} \text{ if } z\notin\{a,b\}, \\
                                                 \bar C^{uv}_{wb} \text{ if } z=a, \\
                                                 \bar C^{uv}_{aw} \text{ if } z=b,
                                   \end{cases} \tag{by~\eqref{disc-eq:rd1} and~\eqref{disc-eq:rd3}} \\
\bar I^{wz};\bar C^{uv}_{ab} &\equiv \bar C^{uv}_{ab};\bar I^{wz}, \tag{by~\eqref{disc-eq:comm3}} \\
\bar C^{wz}_{cd};\bar C^{uv}_{ab} &\phantom{\equiv} \text{ is already in the right form.}
\end{align*}
The inductive case is very similar to that of Lemma~\ref{lma:I-commutes}.
\end{proof}

\begin{corollary}\label{cor:Cbar-form}
Any term is equal to a term of the form $\mathtt I;\mathtt C;\mathtt E^{<0};\mathtt E^{\geq 0};\bar{\mathtt E}^{\geq 0};\bar{\mathtt E}^{<0};\bar{\mathtt C};\mathtt t$, where the term $\mathtt t$ contains no $I$-, $C$-, $E$-, $\bar E$- or $\bar C$-terms.
\end{corollary}

\begin{lemma}\label{lma:Ibar-commutes}
Let $\mathtt t$ be a term not containing any $I$-, $C$-, $E$-, $\bar E$- or $\bar C$-terms such that the term $\mathtt t;\bar I^{uv}$ is defined. Then there exists a term $\mathtt t'$ not containing any $I$-, $C$-, $E$-, $\bar E$- or $\bar C$-terms such that $\mathtt t$ and $\mathtt t'$ have the same number of $\bar I$-terms, and one of the following holds:
\begin{enumerate}[label={(\arabic*)}]
\item $\mathtt t;\bar I^{uv}\equiv\mathtt t'$, or
\item there is a connection $\bar I^{ab}$ such that $\mathtt t;\bar I^{uv}\equiv\bar I^{ab};\mathtt t'$.
\end{enumerate}
\end{lemma}
\begin{proof}
By induction on $\mathtt t$.
\paragraph{Base cases:}
\begin{align*}
\id;\bar I^{uv} &\equiv\bar I^{uv};\id, \\
S^w;\bar I^{uv} &\equiv \bar I^{uv};S^w, \tag{by~\eqref{disc-eq:sd1}} \\
R^{w\mapsto z};\bar I^{uv} &\equiv \bar I^{uv};R^{w\mapsto z}, \tag{by~\eqref{disc-eq:rd1}} \\
\bar I^{wz};\bar I^{uv} &\phantom{\equiv} \text{ is already in the right form.}
\end{align*}
The inductive case is very similar to that of Lemma~\ref{lma:I-commutes}.
\end{proof}

\begin{corollary}\label{cor:Ibar-form}
Any term is equal to a term of the form $\mathtt I;\mathtt C;\mathtt E^{<0};\mathtt E^{\geq 0};\bar{\mathtt E}^{\geq 0};\bar{\mathtt E}^{<0};\bar{\mathtt C};\bar{\mathtt I};\mathtt t$, where the term $\mathtt t$ contains only $S$-, $R$-, and identity terms.
\end{corollary}

\begin{lemma}\label{lma:S-commutes}
Let $\mathtt t$ be a term containing only $S$-, $R$-, and identity terms such that the term $S^{u};\mathtt t$ is defined. Then there exists a term $\mathtt t'$ containing only $S$-, $R$-, and identity terms such that $\mathtt t$ and $\mathtt t'$ have the same number of $S$-terms, and one of the following holds:
\begin{enumerate}[label={(\arabic*)}]
\item $S^{u};\mathtt t\equiv\mathtt t'$, or
\item there is a term $S^{v}$ such that $S^{u};\mathtt t\equiv\mathtt t';S^{v}$.
\end{enumerate}
\end{lemma}
\begin{proof}
By induction on $\mathtt t$.
\paragraph{Base cases:}
\begin{align*}
S^u;\id &\equiv \id;S^u, \\
S^u;S^w &\phantom{\equiv} \text{ is already in the right form,} \\
S^u;R^{w\mapsto z} &\equiv \begin{cases} R^{w\mapsto z};S^u \text{ if } u\neq w, \\
                                    R^{u\mapsto z} \text{ if } u=w.
                      \end{cases} \tag{by~\eqref{disc-eq:sr1} and~\eqref{disc-eq:sr2}}
\end{align*}
The inductive case is very similar to that of Lemma~\ref{lma:I-commutes}.
\end{proof}

\begin{corollary}[Proposition~\ref{prop:ICE-form}]\label{cor:ICE-form}
Any term is equal to a term in an $ICE$-form.
\end{corollary}

\end{document}